\documentclass[10pt,twocolumn,twoside]{IEEEtran}
\usepackage{graphics}
\usepackage{graphicx}
\usepackage{subfigure}
\usepackage{cite}
\usepackage{epsfig}
\usepackage{CJK}
\usepackage{amssymb,amsmath}
\usepackage{amsthm}
\usepackage{lineno}
\usepackage{float}

\newtheorem{theorem}{\textbf{Theorem}}
\newtheorem{myDef}{\textbf{Definition}}
\newtheorem{lemma}{\textbf{Lemma}}
\newtheorem{myremark}{\textbf{Remark}}

\ifCLASSINFOpdf
\else
\fi

\allowdisplaybreaks[1]

\begin{document}
%
\title{Resilient Consensus Through\\ Event-based Communication}
%
%
%

\author{Yuan~Wang
        and~Hideaki~Ishii,~\IEEEmembership{Senior Member,~IEEE}
\thanks{Y.~Wang and H.~Ishii are with the Department
of Computer Science, Tokyo Institute of Technology, Yokohama, Japan.

E-mails: wang.y.bb@m.titech.ac.jp, ishii@c.titech.ac.jp.

This work was supported in the part by the 
JST CREST Grant No.~JPMJCR15K3 and by JSPS under Grant-in-Aid for 
Scientific Research Grant No.~18H01460. The support provided by 
the China Scholarship Council is also acknowledged.}}

\maketitle
\begin{abstract}
We consider resilient versions of discrete-time multi-agent consensus 
in the presence of faulty or even malicious agents in the network. 
In particular, we develop event-triggered update rules which can 
mitigate the influence of the malicious agents and at the same time 
reduce the communication. Each regular agent updates its 
state based on a given rule using its neighbors' information. 
Only when the triggering condition is satisfied, the regular agents 
send their current states to their neighbors. Otherwise, the
neighbors will continue to use the state received the last time.
Assuming that a bound on the number of malicious nodes is known, 
we propose two update rules with event-triggered communication. 
They follow the so-called mean subsequence reduced (MSR) type 
algorithms and ignore values received from potentially malicious 
neighbors. We characterize the necessary
connectivity in the network for the algorithms to perform correctly,
which are stated in terms of the
notion of graph robustness. A numerical example is provided to 
demonstrate the effectiveness of the proposed approach.
\end{abstract}


%
\IEEEpeerreviewmaketitle

\section{Introduction}
\label{S:1}

The study of distributed coordination in multi-agent 
systems has received much attention in a wide range of areas 
including control, robotics, communications, complex networks, and computer science.
More recently, it has been recognized that cyber security 
for networked control systems is a critical issue since
the extensive use of communications for the interactions 
among agents creates 
numerous vulnerabilities for potential attacks (e.g., \cite{Sandberg2015}). 
Control related applications such as those in robotics involve physical aspects,
and hence, different from cyber attacks 
purely in the domain of information technology, 
attacks may lead to damages in equipments or even accidents.

In large-scale multi-agent systems, consensus problems form one of the fundamental 
problems (e.g., \cite{Mesbahi2010}). There, agents interact locally and exchange their information 
with each other in order to arrive at the global objective of sharing a 
common value. In an uncertain environment where faults or even adversarial
attacks can be present, it is of great importance to defend consensus algorithms 
by raising their security levels so as to avoid being influenced 
by such uncertainties in their decision makings. In this context, 
adversarial agents are those that do not follow the given algorithms 
and might even attempt to keep the nonfaulty, regular agents 
from reaching consensus.
It is also remarked that, as a different class 
of cyber attacks, the effects of jamming and DoS
attacks on multi-agent consensus have recently 
been analyzed in \cite{Kikuchi2017,Senejohnny2018}.

In this paper, we study resilient versions of consensus algorithms 
and specifically follow the line of research of
fault-tolerant distributed algorithms in the area of computer science, 
where such problems have long been studied 
(see, e.g., \cite{Defago2016,KieAza94,Lynch1996}).
For each regular agent, a simple but effective approach to mitigate the 
influence of potentially misleading information due to faults and 
cyber attacks is to ignore the agents 
whose states are the most different from its own. It is assumed that 
the nodes know a priori the maximum number $F$ of adversarial agents 
in the network. Hence, it is useful to remove the $F$ largest values 
as well as the $F$ smallest values among those received from the neighbors. 
This class of algorithms is sometimes called the mean subsequence reduced 
(MSR) algorithms and has been employed in computer science 
(e.g., \cite{Mendes2015,Vaidya2012}), 
control theory (e.g., \cite{Dibaji2017,LeBlanc2013,Zhang2015}),
and robotics (e.g., \cite{Guerrero-Bonilla2017,Park2017,Saldana2017}). 
An important recent progress lies in the characterization of the 
necessary requirement on the topology of the agent networks. 
This was initiated by \cite{LeBlanc2013,Vaidya2012}, where
the relevant notion of robust graphs was proposed. 

In this paper, we develop distributed protocols for resilient consensus 
with a particular emphasis on the communication loads for node interactions. 
We reduce the transmissions in MSR algorithms
through the so-called event-triggered protocols 
(e.g., \cite{Heemels2012}).
Event-based protocols have been developed for conventional consensus 
without malicious agents in, e.g., 
\cite{Dimarogonas2012,Guo2014,Kia2015,Ma2017,MengChen2013,Meng2017,Seyboth2013}.
Related results can be found in \cite{Kadowaki2015}, where event-based
consensus-type algorithms are developed for the synchronization of 
clocks possessed by the nodes in wireless sensor networks (WSNs).

Under this method, nodes make transmissions only when necessary 
in the sense that their values sufficiently changed since their 
last transmissions. In certain cases, the agents may make only
a finite number of transmissions to neighbors. 
The advantage is that the communication can be greatly reduced 
in frequency and may be required only a finite number of times, 
while the tradeoff is that the achievable level 
of consensus may be limited, leaving some gaps in the agents' values.
Time-triggered protocols may be a simpler way to reduce the
communication load, but will not be able to determine when 
to stop the communication.
More concretely, we develop two protocols for resilient consensus 
under event-based communication. Their convergence properties 
are analyzed, and the requirement for the network topology is 
fully characterized in terms of robust graphs. We will show 
through a numerical example how the two protocols differ in 
the amounts of communication needed for achieving consensus.

The difficulty in applying event-triggered protocols in the context 
of resilient consensus based on MSR algorithms
is due to the handling of the errors 
between the current values and their last transmitted ones. 
In our approach, we treat such errors as noise in the system. 
This approach can be seen as an extension of \cite{Kikuya2018},
where a resilient version of the WSN clock synchronization problem 
in \cite{Kadowaki2015}
mentioned above is analyzed; the exchange of two clock variables 
creates decaying noises in the consensus-type algorithms. 
By contrast, in our problem setting, the errors are due to triggering 
and do not entirely decay to zero. Moreover, we study a different class of adversarial 
nodes as we clarify later. 

Another feature of this paper is that we deal with event-triggered
protocols for consensus algorithms in the discrete-time domain. 
This is in contrast to the conventional works that
deal with event-based consensus in continuous time
(e.g., \cite{Dimarogonas2012,Kia2015,Ma2017,Seyboth2013}).
In such cases, the agents must continuously monitor their states 
to detect when their states reach the thresholds for
triggering events. This mechanism may require special resources for computation. 
Furthermore, events with short intervals may occur, which can
result in undesirable Zeno behaviors. 
On the other hand, there are works such as \cite{Guo2014,MengChen2013,Meng2017},
where sampled-data controllers are employed for agents with
system dynamics in continuous time.

It is interesting to note that in discrete time, 
event-based consensus algorithms must be designed differently. This issue
has also been discussed in \cite{Kadowaki2015}, which essentially
deals with discrete-time asynchronous update rules without adversaries.
It is emphasized that in the presence of attacks, this aspect seems 
even more crucial.
In this paper, we present two resilient consensus algorithms,
but also discuss a third potential approach. 
The differences among them are modest: At the updates, each 
agent has the option of using its own state or its own last 
transmitted state. We will however see that analysis methods 
can differ, leading to various levels of conservatism 
in the bounds on the parameters for the event triggering functions.

Before we close this introduction, 
we would like to briefly discuss the recent advances 
in the research of MSR algorithms.
The early works \cite{LeBlanc2013,Vaidya2012} dealt with first-order 
agents with synchronous updates. 
In \cite{Dibaji2017}, 
MSR-type algorithms are developed
for agents having second-order dynamics, which may hence be
applicable to autonomous vehicles, 
and moreover, delays in communication 
as well as asynchronous updates are taken into account. 
The work \cite{LeBlanc2017} studied the MSR-based resilient 
synchronization problem in a more general setting with agents 
having higher-order dynamics, operating
in continuous time.
While most studies mentioned so far deal with agents whose
states take real values, the work \cite{Dibaji2018} considers
agents with quantized (i.e., integer-valued) states.
Also, there is a line of graph theoretic studies 
(e.g., \cite{Usevitch2017,Zhang2015,Zhao2017}),
which discuss methods to identify the robustness of certain classes 
of graphs with specified levels of robustness, for both undirected 
and directed graphs.



The MSR-based approach is found useful in addressing distributed
problems outside of consensus problems as well. 
In \cite{Sundaram2018}, a resilient version of 
distributed optimization is studied by employing MSR-like 
mechanisms to detect outliers in the neighbors' variables. 
Further, in \cite{Mitra2018}, resilient distributed state estimation problem 
is studied, where another class of robust graphs relevant to the problem
is introduced. 
In the robotics area, \cite{Guerrero-Bonilla2017} applies MSR algorithms 
for cooperative robots and develops methods for the robots to find
if and how the network for their interactions can be built 
with robust graph properties. 

It should highlight that MSR algorithms do not aim to detect 
adversarial nodes as they simply leave out 
the values that are the most different. Efforts have been made to 
develop distributed algorithms to detect and identify the adversaries 
while performing the given task. For example, 
in \cite{pasqualetti2012,sundaram2011}, detection techniques
using unknown input observers are proposed, in which case 
the initial values of all normal nodes can be found though 
they require the global knowledge of the network topology. 
Methods based on more local information can be found in, e.g., 
\cite{Chen2018,guo2012,zhao2018}.

This paper is organized as follows. In Section~\ref{S:2}, 
we introduce some preliminaries and then formulate the event-based 
resilient consensus problem. We propose two event-based resilient 
update rules and study their convergence and necessary network 
structures in Sections~\ref{S:3} and~\ref{S:4}. 
A numerical example is given in Section~\ref{S:5} to demonstrate 
the effectiveness of the proposed algorithms. 
We provide concluding remarks in Section~\ref{S:6}.
This paper is an extended version of \cite{Wang2018} with 
full proofs of the results and further discussions.

\section{Event-based resilient consensus problem}
\label{S:2}

\subsection{Preliminaries on graphs}
Some basic notations related to graphs are introduced 
for the analysis in this paper.

Consider the directed graph $\mathcal{G}=(\mathcal{V},\mathcal{E})$ consisting of $n$ nodes. Here the set of nodes is denoted by $\mathcal{V}=\{1,2,\ldots,n\}$ and the edge set by $\mathcal{E} \subseteq \mathcal{V}\times\mathcal{V}$. The edge $(j,i)\in \mathcal{E}$ indicates that node $j$ can send a message to node $i$ and is called an incoming edge of node $i$. 
Let $\mathcal{N}_{i}=\{j:(j,i)\in \mathcal{E}\}$ be the set 
of neighbors of node $i$. The number of neighbors of node~$i$ 
is called its degree and is denoted as $d_i = |\mathcal{N}_i|$
The path from node $i_1$ to node $i_p$ is denoted as the sequence $(i_1,i_2,\ldots,i_p)$, where $(i_j,i_{j+1})\in \mathcal{E}$ for $j=1,2,\ldots,{p-1}$. The graph $\mathcal{G}$ is said to have a spanning tree if there exists a node from which there is a path to all other nodes of this graph.

To establish resilient consensus results, an important topological notion is that of robustness of graphs \cite{LeBlanc2013}.

\begin{myDef} \label{Def1}\rm
The graph $\mathcal{G}=(\mathcal{V},\mathcal{E})$ is called 
\emph{$(r,s)$-robust} ($r,s < n$) if for any two nonempty disjoint subsets $\mathcal{V}_{1},\mathcal{V}_{2}\subseteq\mathcal{V}$, one of 
the following conditions is satisfied:
1)~$\mathcal{X}^{r}_{\mathcal{V}_{1}}=\mathcal{V}_{1}$,
2)~$\mathcal{X}^{r}_{\mathcal{V}_{2}}=\mathcal{V}_{2}$, and
3)~$|\mathcal{X}^{r}_{\mathcal{V}_{1}}|+|\mathcal{X}^{r}_{\mathcal{V}_{2}}|\geq s$,
where $\mathcal{X}^{r}_{\mathcal{V}_{i}}$ is the set of all nodes in $\mathcal{V}_{i}$ which have at least $r$ neighbors outside $\mathcal{V}_{i}$ for $i=1,2$.
The graph is said to be $r$-robust if it is $(r,1)$-robust. 
\end{myDef}

\begin{figure}[t]
\centering
\includegraphics[width=0.4\linewidth]{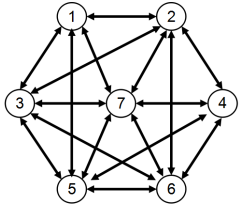}
\vspace*{-2mm}
\caption{Network topology with $(3,3)$-robustness}
\label{fig2}
\vspace*{-4mm}
\end{figure}

In Fig.~\ref{fig2}, we display an example graph 
with seven nodes. 
It can be checked to have just enough connectivity 
to be (3,3)-robust. 
This level of robustness is lost 
if any edge is removed.

We summarize some basic properties of robust graphs 
\cite{LeBlanc2013}. 

\begin{lemma} \label{lemma1}\rm
An $(r,s)$-robust graph $\mathcal{G}$ satisfies the following:
\begin{enumerate}
  \item $\mathcal{G}$ is $(r',s')$-robust, where 
        $0 \le r' \le r$, $1 \le s' \le s$, and 
        in particular, it is $r$-robust.
  \item $\mathcal{G}$ has a directed spanning tree. 
        Moreover, it is $1$-robust if and only if it 
        has a directed spanning tree.
  \item $r \le \lceil n/2 \rceil$, 
        where the ceil function $\lceil y \rceil$ 
        gives the smallest integer greater than or equal to $y$.
        Furthermore, it holds 
        $r = \lceil n/2 \rceil$ if and only if $\mathcal{G}$ is a 
        complete graph.
  \item The degree $d_i$ for $i \in \mathcal{V}$ is lower bounded 
        as $d_i \ge r+s-1$ if $s<r$ and $d_i \ge 2r-2$ if $s \ge r$.
\end{enumerate}
Moreover, a graph $\mathcal{G}$ is $(r,s)$-robust if it is $(r+s-1)$-robust.
\end{lemma}

In consensus problems, the property 2) in the lemma is 
of interest. Robust graphs may not be strongly connected in general,
but this property indicates that the notion of robust graphs 
is a generalization of graphs
containing directed spanning trees, 
which are of great relevance in the literature of 
consensus \cite{Mesbahi2010}.
As we will see, robust graphs play a key role in characterizing 
the necessary network structure for achieving resilient consensus. 
It should however be noted that checking the robustness of a given 
graph involves combinatorial computation and is thus difficult 
in general \cite{Usevitch2017,Zhang2015,Zhao2017}.

\vspace*{-2mm}
\subsection{Event-based consensus protocol}

We consider the directed graph $\mathcal{G}$ of $n$ nodes. The nodes in $\mathcal{V}$ are partitioned into two sets:  $\mathcal{R}$ denotes the set of regular nodes and $\mathcal{A} = \mathcal{V}\setminus \mathcal{R}$ represents the set of adversarial nodes. The regular nodes will follow the designed algorithm exactly while the adversarial nodes can have different update rules from that of the regular nodes. The attacker is allowed to know 
the states of the regular nodes and
the graph topology, and to choose any node
as a member of $\mathcal{A}$ under some constraints.

We introduce the event-based protocol for the regular nodes to 
achieve consensus. It can be outlined as follows: 
At each discrete-time instant $k \in \mathbb{Z}_+$, 
the nodes make updates, but whether they transmit their 
current values to neighbors depends on the triggering 
function. 
More concretely, each node $i$ has an auxiliary 
variable which is its state value communicated the last 
time and compares it with its own current state. 
If the current state has changed sufficiently, then it 
will be sent to its neighbors and the auxiliary variable 
will be replaced.

The update rule for agent $i$ is described by
\begin{equation}
  {x_i}(k + 1) 
   =   {x_i}(k) + {u_i}(k),
\label{eqn:x_i}
\end{equation}
where $x_i(k) \in \mathbb{R}$ is the state 
and $u_i(k)$ is the control given by
\begin{equation}
  {u_i}(k)
   = \sum\limits_{j \in \mathcal{N}_i }  a_{ij}(k) 
         \left( \hat{x}_j(k) - x_i(k) \right).
\label{eqn:u_i}
\end{equation}
Here, $\hat{x}_j(k)\in \mathbb{R}$ is an auxiliary state,
representing the last communicated state of node $j$ at time $k$.
It is defined as 
\[
  \hat{x}_j(k)=x_j(t_{l}^{j}),~~
   k \in [t_{l}^{j},t_{l+1}^{j}), 
\]
where $t_{0}^{j},t_{1}^{j}, \ldots$ denote the 
transmission times of node $j$ determined by the triggering function to be given below. The initial values $x_i(0)$, $\hat{x}_j(0)$ are given,
and $a_{ij}(k)$ is the 
weight for the edge $(j,i)$.
Also, let 
$a_{ii}(k) 
   = 1 - \sum_{j \in\mathcal{N}_i}
           a_{ij}(k)$.
Assume that
$\gamma \leq a_{ij}(k)<1$ if $a_{ij}(k)\neq 0$
or if $i=j$ for $i,j\in\mathcal{V}$,
where $\gamma$ is the lower bound with 
$0<\gamma\leq 1/2$. 
In the resilient consensus algorithms to be introduced, 
the neighbors 
whose values are used for updates change over time,
and hence, the weights $a_{ij}(k)$ are time varying.
The update rule above can be seen as a discrete-time 
counterpart of the event-based consensus algorithms 
in, e.g., \cite{Guo2014,Ma2017,Seyboth2013}.

We now introduce the triggering function. 
Denote the error at time $k$ between the updated state $x_i(k+1)$ 
and the auxiliary state $\hat{x}_i(k)$ 
by $e_i(k) = \hat{x}_i(k) - x_i(k+1)$ for $k\geq 0$.
Then, let
\begin{equation}
  {f_i}(k) 
    = \left| 
        {{e_i}(k)} 
      \right| 
     - \left( 
         {{c_0} + {c_1}{\text{e}^{ - \alpha k}}} 
       \right),
\label{eq-02}
\end{equation}
where $c_{0}$, $c_{1}$, and $\alpha>0$ are positive constants.
If $f_i(k)>0$, agent~$i$ transmits its new state 
$x_i(k+1)$ to the neighbors at time $k$. 
This mechanism will be discussed further later. 

\vspace*{-2mm}
\subsection{Adversary model and resilient consensus}

For each adversarial node $i$ in the set $\mathcal{A}$, 
its state $x_i(k)$ is updated as in 
\eqref{eqn:x_i}, but its control $u_i(k)$ can take
arbitrary values at any $k$. Such nodes may have knowledge
on the entire network including its topology, the values of
all normal nodes, and their update rules. In this respect,
we take account of their worst-case behaviors. 
Specifically, we employ the malicious model 
introduced in \cite{LeBlanc2013} as follows:

\begin{myDef} \label{Def2}\rm
(\emph{Malicious nodes}): An adversarial 
node~$i$ is said to be malicious if it updates
its state $x_i(k)$ in \eqref{eqn:x_i}
by arbitrarily choosing $u_i(k)$ and
sends the state $x_i(k)$ to all of its neighbors 
at each transmission.
\end{myDef}

Adversarial nodes more difficult to deal with are 
those that can transmit different values to different neighbors.
Such nodes are referred to as being
Byzantine \cite{Vaidya2012}. The motivation for considering 
malicious nodes as defined above comes, for example, 
from the applications of WSNs, 
where sensor nodes communicate to their neighbors by broadcasting their data. 

We also set a bound on the number of malicious nodes in the network. 
In this paper, we will deal with networks of the so-called $F$-total 
model as defined below. 

\begin{myDef} \label{Def3}\rm
(\emph{$F$-total model}): 
For $F \in \mathbb{N}$, we say that the adversarial set $\mathcal{A}$ 
follows an \emph{$F$-total} model 
if $|\mathcal{A}|\leq F$.
\end{myDef}

Let the number of malicious agents be denoted by $N_m=|\mathcal{A}|$. 
Then, let $N=|\mathcal{V}|-N_m$ be the number of regular agents. 

Now, we introduce the notion of resilient consensus for multi-agent systems.

\begin{myDef} \label{Def4}\rm
 (\emph{Resilient consensus}): 
Given $c\geq 0$,
if for any possible sets and behaviors of the malicious agents
 and any initial state values of the regular nodes, the following conditions are satisfied, then the multi-agent system is said to reach resilient consensus
at the error level $c$:
\begin{enumerate}
  \item Safety condition: There exists an interval $\mathcal{S} \subset \mathbb{R}$ such that  $x_i(k)\in \mathcal{S}$ for all $ i \in \mathcal{R}$, $k\in \mathbb{Z}_{+}$.
  \item Consensus condition: 
  For all $i, j \in \mathcal{R} $, it holds 
  that $\mathop {\limsup }_{k \to \infty } \left| x_i(k) - x_j(k) \right| \le c$.
\end{enumerate}
\end{myDef}

In this paper, we would like to design event-based 
update rules for the regular agents to reach resilient 
consensus at a prespecified error level $c$
under the $F$-total model
by using only local information obtained from their neighbors

\section{Robust protocols for event-based consensus}
\label{S:3}

\subsection{E-MSR algorithm}

In this section, we outline a distributed protocol to solve 
the resilient consensus problem. As discussed above, every node 
makes an update at every time step in a synchronous manner, 
but only when an event happens, the auxiliary values will 
be updated and then sent to neighbors. The basis of the algorithm follows
those in the works of, e.g., 
\cite{Dibaji2017,LeBlanc2013}. 
The algorithm in this paper is called the 
event-based mean subsequence reduced (E-MSR) algorithm.

The E-MSR algorithm has four steps as follows:
\begin{enumerate}
  \item \emph{(Collecting neighbors' information)} At each time step $k$, every regular node $i\in\mathcal{R}$ uses the values $\hat{x}_j(k), j\in \mathcal{N}_i$, most recently communicated from the neighbors as well as its own value ${x}_i(k)$ and sorts them from the largest to the smallest.
  \item \emph{(Deleting suspicious values)} Comparing with ${x}_{i}(k)$, node $i$ removes the $F$ largest and $F$ smallest values from its neighbors. If the number of values larger or smaller than ${x}_{i}(k)$ is less than \emph{F}, then all of them are removed. The removed data is considered as suspicious and will not be used in the update. The set of the node indices of the remaining values is written as $\mathcal{M}_{i}(k)\subset \mathcal{N}_{i}$.
  \item \emph{(Local update)} Node $i$ updates its state by
  \begin{equation}\label{eq-03}
    {x_i}(k + 1) 
     = {x_i}(k) 
        + \sum\limits_{j \in \mathcal{M}_{i}(k)} a_{ij}(k)
          \left( 
            \hat{x}_j(k) - x_i(k) 
          \right).
  \end{equation}
  \item \emph{(Communication update)} Node $i$ checks if 
     its own triggering function $f_i(k)$ in \eqref{eq-02} is 
     positive or not. Then, it sets $\hat{x}_i(k+1)$ as
    \begin{equation}\label{eq-04}
       \hat{x}_i(k+1) 
          = \begin {cases}
               x_i(k+1) & \text{if}~f_i(k)>0, \\
               \hat{x}_i(k) & \text{otherwise}.
            \end {cases}
    \end{equation}
\end{enumerate}

The communication rule in this algorithm shows that only when the current value has varied enough to exceed a threshold, then the auxiliary variable will be updated, and only at this time the node sends its value to its neighbors. This event triggering scheme can significantly reduce the communication burden as we will see in the numerical example
in Section~\ref{S:5}.

\vspace*{-2mm}
\subsection{Protocol~1} \label{sec:protocol1}
The first protocol of this paper is the E-MSR algorithm 
as stated above, which will be referred to as Protocol~1. 
We are now ready to present our main result for this protocol.

We introduce two kinds of 
minima and maxima of the states of the regular agents:
The first involves only the states as 
$\overline{x}(k) = \mathop {\max }_{i\in\mathcal{R}} x_i(k)$ 
and $\underline{x}(k) = \mathop {\min }_{i\in\mathcal{R}} x_i(k)$
while the second uses also the auxiliary variables as
$\underline{\hat{x}}(k) = \min_{i\in\mathcal{R}} \{x_{i}(k), \hat{x}_{i}(k)\}$
and $\overline{\hat{x}}(k) = \max_{i\in\mathcal{R}} \{x_{i}(k), \hat{x}_{i}(k)\}$.
The safety interval $\mathcal{S}$ is chosen as 
$\mathcal{S} = \left[ \underline{\hat x}(0), \overline{\hat x}(0) \right]$.
It is noted that at initial time, 
$\hat{x}_i(0)$ need not be the same as $x_i(0)$.

\begin{theorem} \label{theorem1}\rm
Under the $F$-total model, the regular agents with E-MSR using 
\eqref{eq-03} and \eqref{eq-04} reach resilient consensus 
at an error level $c$ if 
and only if the underlying graph is $(F+1,F+1)$-robust. 
The safety interval is given by 
$\mathcal{S} = \left[ \underline{\hat{x}}(0), \overline{\hat{x}}(0) \right]$, 
and the consensus error level $c$ is achieved if
the parameter $c_0$ in the triggering function \eqref{eq-02}
satisfies
\begin{equation}
  c_0 
   \leq \frac{\gamma^N c}{4N}.
\label{eqn:thm1}
\end{equation}
\end{theorem}

\vspace*{.5mm}
\begin{proof}
(Necessity)
This essentially follows from \cite{LeBlanc2013}, which considers 
the special case without the triggering function, that is, $c_0=c_1=0$.

(Sufficiency)
We first show that the interval 
$\mathcal{S} = \left[ \underline{\hat{x}}(0), \overline{\hat{x}}(0) \right]$ 
satisfies the safety condition by induction. 
Note that the update rule 
\eqref{eq-03} can be rewritten as
\begin{equation} \label{eq-n01}
  x_i(k+1) 
    = a_{ii}(k){x_i}(k)
       + \sum\limits_{j \in {\mathcal{M}_i}(k)} 
             a_{ij}(k) \hat{x}_j(k),
\end{equation}
where
$a_{ii}(k) 
   = 1 - \sum_{j \in {\mathcal{M}_i}(k)} 
           a_{ij}(k)$.
%
At time $k=0$, it is clear by definition that 
$x_i(0),\hat{x}_i(0)\in \mathcal{S}$, $i \in \mathcal{R}$. 
Suppose that for each regular agent~$i$, $x_i(k),\hat{x}_i(k)\in \mathcal{S}$. 
Then, for agent~$i$, its neighbors in $\mathcal{M}_i(k) $ 
take values only in $\mathcal{S}$, since there are agents
with values outside $\mathcal{S}$ at most $F$, and they are
ignored in step~2 of the E-MSR. 
From \eqref{eq-n01}, we have 
${x_i}(k+1) \in \mathcal{S}$. 
Moreover, by \eqref{eq-04}, it follows that
$\hat{x}_i(k+1) \in \mathcal{S}$. 
Thus, $\mathcal{S}$ is the safety interval. 

We next establish the consensus condition. 
Note that for time $k \in (t_{l}^{i},t_{l+1}^{i})$ between 
two triggering instants, we have $f_i(k)\leq 0$. 
Moreover, for the neighbor node $j \in \mathcal{N}_i$, 
if $f_j(k)>0$, then we have $\hat x_j(k+1) = x_j(k+1)$. 
If $f_j(k) \leq 0$, then 
$\hat x_j(k+1) = \hat x_j(k) = x_j(k+1)+{e_j}(k)$. 
As a result, it holds 
$\hat x_j(k) =  x_j(k)+\hat{e}_j(k-1)$ for $k\geq 1$, 
where
\begin{equation*}
  \hat{e}_j(k) 
    = \begin {cases}
        e_j(k) & \text{if}\ f_j(k) \leq 0, \\
        0 &  \text{otherwise}.
      \end{cases} 
\end{equation*}
Note that
\begin{equation}\label{eq-**01}
  |\hat{e}_j(k)| 
   \le c_0 + c_1\text{e}^{-\alpha k},\ \ \forall k \ge 0.
\end{equation}
Then, we can write \eqref{eq-n01} as
\begin{flalign}\label{eq-06}
  x_i(k + 1)
  & = a_{ii}(k){x_i}(k) 
        + \sum\limits_{j \in \mathcal{M}_i(k) } 
                 a_{ij}(k) \left(x_j(k) + \hat{e}_j(k-1) \right).
\end{flalign}
This can be bounded by using the maximum state $\overline{x}(k)$ as
\begin{flalign}\label{eq-08}
 x_i(k + 1)
  &\le a_{ii}(k)\overline x (k) 
       + \sum\limits_{j \in {\mathcal{\mathcal{M}}_i}(k)} 
                 a_{ij}(k) \left(\overline x (k)  + \hat{e}_j(k-1) \right)
      \nonumber \\
  &= \overline x (k) 
        + \sum\limits_{j \in {\mathcal{M}_i}(k)} 
           a_{ij}(k) \hat{e}_j(k-1)
      \nonumber \\
  & \le \overline{x}(k) 
         + \max \limits_{j\in \mathcal{M}_i(k)} 
              \left| \hat{e}_j(k-1) \right|.
\end{flalign}
Thus, by \eqref{eq-**01} it follows
\begin{equation*}
  x_i(k + 1) 
    \le \overline x (k) + c_0
          +  c_1 \text{e}^{ - \alpha (k-1)}. 
\end{equation*}

Let $V(k) = \overline{x}(k) - \underline{x}(k)$. 
Then, introduce two sequences by
\begin{align}
  \overline{x}_0(k+1)
    = \overline{x}_0(k) + c_0 + c_1\text{e}^{-\alpha (k-1)},
       \label{eq-*1}\\
   \underline{x}_0(k+1)
    =\underline{x}_0(k) - c_0 - c_1\text{e}^{-\alpha (k-1)},
       \label{eq-*2}
\end{align}
where $\overline{x}_0(0) = \overline{x}(0) - \sigma _0$, 
and $\underline{x}_0(0) = \underline{x}(0) + \sigma _0$ 
with $\sigma _0 = \sigma V(0)$.
We next introduce another sequence $\varepsilon _0(k)$ defined by
\begin{align}\label{eq-*3}
  \varepsilon _0(k + 1) 
    = \gamma \varepsilon_0(k) - (1 - \gamma )\sigma_0,
\end{align}
where $\varepsilon _0(0) = \varepsilon V(0)$. 
Take the parameters $\varepsilon$ and $\sigma$ so that
\begin{equation}\label{eq-I01}
   \varepsilon + \sigma = \frac{1}{2},~~ 
    0 < \sigma  < \frac{\gamma ^N}{1 - \gamma^N}\varepsilon.
\end{equation}
For the sequence $\varepsilon_0(k)$, let
\begin{align*}
  \overline{\mathcal{X}_0}(k, \varepsilon _0(k)) 
   &= \left\{ 
        j \in \mathcal{V}:\ 
          x_j(k) > \overline{x}_0(k) - \varepsilon _0(k)
     \right\}, \\
  \underline{\mathcal{X}_0}(k, \varepsilon _0(k)) 
   &= \left\{ 
        j \in \mathcal{V}:\
         x_j(k) < \underline{x}_0(k) + \varepsilon _0(k)
      \right\}. 
\end{align*}
These two sets are both nonempty at time $k=0$ and, in particular, 
each contains at least one regular node; this is because 
by definition, $\overline{x}(0) > \overline{x}_0(0) - \varepsilon_0(0)$ 
and $\underline{x}(0) < \underline{x}_0(0) + \varepsilon_0(0)$.

In the following, we show that 
$\overline{ \mathcal{X}_0}(k,\varepsilon _0(k))$ and 
$\underline{ \mathcal{X}_0}(k,\varepsilon _0(k))$ are 
disjoint sets. To this end, we must show
\[
   \overline{x}_0(k) - \varepsilon_0(k) 
     \ge \underline{x}_0(k) + \varepsilon_0(k).
\]
By \eqref{eq-*1} and \eqref{eq-*2} for
$\overline{x}_0(k)$ and $\underline{x}_0(k)$, we have
\begin{align}\label{eq-ex01}
  &(\overline{x}_0(k) - \varepsilon_0(k)) 
     - (\underline{x}_0(k) + \varepsilon _0(k)) 
   \nonumber \\
  &  {\kern 10pt} 
     = \left( 
       \overline{x}_0(0) + c_0 k 
         + c_1 \frac{1 - \text{e}^{-\alpha (k-1)}}{%
                       1 - \text{e}^{-\alpha}}
       \right) \nonumber \\
  &  {\kern 25pt} 
      - \left( 
          \underline{x}_0(0) - c_0 k 
         - c_1 \frac{1 - \text{e}^{-\alpha (k-1)}}{%
                     1 - \text{e}^{-\alpha }}
         \right) - 2 \varepsilon_0(k).
\end{align}
Then by substituting 
$\overline{x}_0(0) = \overline{x}(0) - \sigma_0$ and 
$\underline{x}_0(0) = \underline{x}(0) + \sigma_0$ into 
the right-hand side of \eqref{eq-ex01}, we obtain
\begin{align}\label{eq-ex02}
  & \left(
       \overline{x}_0(k) - \varepsilon_0(k)
    \right) 
     - \left(
         \underline{x}_0(k) + \varepsilon_0(k)
       \right)  \nonumber  \\
  & = \left( 
         \overline{x}(0) - \underline{x}(0)
      \right) 
     - 2 \sigma_0 + 2 c_0 k 
         + 2 c_1 \frac{1 - \text{e}^{-\alpha (k-1)}}{%
                       1 - \text{e}^{-\alpha }} - 2 \varepsilon_0(k)
    \nonumber \\
  & = V(0) - 2\sigma V(0) + 2 c_0 k 
       + 2 c_1 \frac{1 - \text{e}^{-\alpha (k-1)}}{%
                     1 - \text{e}^{-\alpha }} - 2 \varepsilon_0(k).
\end{align}
By \eqref{eq-*3} and $0<\gamma\leq 1/2$, we easily have that 
$\varepsilon_0(k + 1) < \varepsilon_0(k)$, and hence
$\varepsilon_0(k)<\varepsilon_0(0)=\varepsilon V(0)$.
We thus obtain
\begin{align*}
 &(\overline{x}_0(k) - \varepsilon_0(k)) 
    - (\underline{x}_0(k) + \varepsilon_0(k)) \\
 &\hspace*{5mm}
   > (1 - 2\sigma  - 2\varepsilon )V(0) + 2 c_0 k 
     + 2 c_1 \frac{1 - \text{e}^{-\alpha (k-1)}}{%
                   1 - \text{e}^{-\alpha }} 
   > 0, 
\end{align*}
where the last inequality holds since 
$\sigma + \varepsilon = 1/2$ from \eqref{eq-I01}. 
Thus, 
$\overline{\mathcal{X}}_0(k,\varepsilon_0(k))$ and 
$\underline{\mathcal{X}}_0(k,\varepsilon_0(k))$ are disjoint sets.

From the above, we have that the two sets 
$\overline{\mathcal{X}_0}(0,\varepsilon_0(0))$ and 
$\underline{\mathcal{X}_0}(0,\varepsilon_0(0))$ are 
nonempty with at least one regular node in each and moreover 
disjoint. Therefore, by the assumption of $(F+1,F+1)$-robustness, 
there are three cases:
\begin{enumerate}
  \item All nodes in $\overline{\mathcal{X}_0}(0,\varepsilon_0(0))$ 
        have $F+1$ neighbors or more from outside.
  \item All nodes in $\underline{\mathcal{X}_0}(0,\varepsilon_0(0))$ 
        have $F+1$ neighbors or more from outside.
  \item The total number of nodes in $\overline{\mathcal{X}_0}(0,\varepsilon_0(0))$ 
        and $\underline{\mathcal{X}_0}(0,\varepsilon_0(0))$ having $F+1$ neighbors 
        or more from outside of its own set is no smaller than $F+1$.
\end{enumerate}
Notice that in any of the three cases, there exists at least one regular 
agent $i \in \mathcal{R}$ in either $\overline{\mathcal{X}_0}(0,\varepsilon_0(0))$ 
or $\underline{\mathcal{X}_0}(0,\varepsilon_0(0))$ that has $F+1$ neighbors 
or more from outside of its own set. 
Here, suppose that this node $i$ belongs to 
$\overline{\mathcal{X}_0}(0,\varepsilon_0(0))$. 
A similar argument holds for the case when it is 
in $\underline{\mathcal{X}_0}(0,\varepsilon_0(0))$.

Now, we go back to \eqref{eq-06} and rewrite it by partitioning 
the neighbor node set $\mathcal{M}_{i}(k)$ of node $i$ into two parts: 
The nodes which belong to 
$\overline{\mathcal{X}}_0(k,\varepsilon_0(k))$ and 
those that do not. Since node $i$ has at least $F+1$ neighbors
outside $\overline{\mathcal{X}}_0(k,\varepsilon_0(k))$, 
the latter set is nonempty. Hence, we obtain
\begin{align*}
 &x_i(k + 1) \\
 &~~= a_{ii}(k) x_i(k) 
     + \sum\limits_{j \in \mathcal{M}_i(k) 
       \cap {\overline {\mathcal{X}}}_0 } a_{ij}(k) x_j(k)\\
 &\hspace*{8mm} \mbox{}
    + \sum\limits_{j \in \mathcal{M}_i(k)\setminus \overline{\mathcal{X}}_0}
         a_{ij}(k) x_j(k) 
    + \sum\limits_{j \in \mathcal{M}_i(k)} 
         a_{ij}(k) \hat{e}_j(k-1),
\end{align*}
where we use the shorthand notation $\overline{\mathcal{X}}_0$ 
for $\overline{\mathcal{X}}_0(k,\varepsilon_0(k))$.
Then, we can bound this from above as
\begin{align}\label{eq-16}
  x_i(k + 1) 
  &\le a_{ii}(k) \overline x (k) 
        + \sum\limits_{j \in \mathcal{M}_i(k) \cap\overline{\mathcal{X}}_0} 
          a_{ij}(k) \overline{x}(k) 
    \nonumber \\
&\hspace*{8mm} \mbox{}
    + \sum\limits_{j \in {\mathcal{M}_i }(k)\setminus\overline{\mathcal{X}}_0} 
         a_{ij}(k)
      \left( 
         \overline{x}_0(k) - \varepsilon_0(k)
      \right) \nonumber\\
&\hspace*{15mm} \mbox{}
    + \sum\limits_{j \in \mathcal{M}_i(k)} a_{ij}(k)\hat{e}_j(k-1)
   \nonumber \\
&= \left( 
      1 - \sum\limits_{j \in \mathcal{M}_i(k)\setminus \overline{\mathcal{X}}_0}
       a_{ij}(k)
    \right)\overline x (k) 
  \nonumber \\
&\hspace*{8mm} \mbox{} 
    + \sum\limits_{j \in \mathcal{M}_i(k) \setminus \overline{\mathcal{X}}_0} 
       a_{ij}(k)
      \left(
          \overline{x}_0(k) 
           - \varepsilon_0(k) 
      \right)  \nonumber \\
&\hspace*{15mm} \mbox{} 
   + \sum\limits_{j \in \mathcal{\mathcal{M}}_i(k)} 
        a_{ij}(k) \hat{e}_j(k-1).
\end{align}

We next show that $\overline{x}(k) \le \overline{x}_0(k) + \sigma_0$ 
(and similarly, $\underline{x}(k) \ge \underline{x}_0(k) - \sigma_0$) 
by induction.
For $k=0$, by definition, we have $\overline{x} (0) = {\overline{x} _0}(0) + {\sigma _0}$.
Suppose that $\overline{x} (k) \le {\overline{x} _0}(k) + {\sigma _0}$.
Then, from \eqref{eq-08} and \eqref{eq-*1}, we have
\begin{align*}
  \overline{x} (k + 1)
   &\le \overline{x} (k) 
     + \mathop {\max }\limits_j 
      \left| {{\hat e_j}(k-1)} \right| 
    \le \overline{x} (k) + {c_0} + {c_1}{\text{e}^{ - \alpha (k-1)}} \\
  &\le {\overline{x} _0}(k) + {\sigma _0} + {c_0} 
     + {c_1}{\text{e}^{ - \alpha (k-1)}} 
  = {\overline{x} _0}(k + 1) + {\sigma _0}. 
\end{align*}

Then, \eqref{eq-16} can be further bounded as
\begin{align}\label{eq-18}
 &{x_i}(k + 1)
 \le \left( 
         {1 - \sum\limits_{j \in {\mathcal{M}_i }(k)\setminus {\overline {\mathcal{X}}}_0} {{a_{ij}}(k)} } 
       \right)
     \left( {{{\overline{x} }_0}(k) + {\sigma _0}} \right) 
       \nonumber \\
&\hspace*{18mm}\mbox{} 
   + \sum\limits_{j \in {\mathcal{M}_i }(k)\setminus {\overline {\mathcal{X}}}_0} {{a_{ij}}(k)} 
      \left( {{{\overline{x} }_0}(k) - {\varepsilon _0}(k)} \right) 
    \nonumber\\
&\hspace*{25mm}\mbox{} 
   + \sum\limits_{j \in {\mathcal{M}_i}(k)} {{a_{ij}}(k)}  
       {{\hat{e}_j}(k-1)} \nonumber  \\
&\le {\overline{x} _0}(k) 
     + \left( 
         {1 - \sum\limits_{j \in {\mathcal{M}_i }(k)\setminus {\overline {\mathcal{X}}}_0} {{a_{ij}}(k)} } 
       \right){\sigma _0}  
      \nonumber \\
&\hspace*{5mm}\mbox{}  
    -  \sum\limits_{j \in {\mathcal{M}_i }(k)\setminus {\overline {\mathcal{X}}}_0} {{a_{ij}}(k){\varepsilon _0}(k)  }
  +  \sum\limits_{j \in {\mathcal{M}_i}(k)} {{a_{ij}}(k)} 
        {\left| {{\hat{e}_j}(k-1)} \right| } .
\end{align}

We also show that $\varepsilon_{0}(k)>0$ holds for $k= 0,1,\ldots,N$. 
It is clear 
from \eqref{eq-*3} that
$\varepsilon_{0}(k+1)<\varepsilon_{0}(k)$. Thus 
we only need to guarantee $\varepsilon_{0}(N)>0$. 
By \eqref{eq-*3}, $\varepsilon_0(N)$ can be written as
\begin{align*}
 {\varepsilon _0}(N)
  &= {\gamma ^N}{\varepsilon _0}(0) 
     - \sum\limits_{i = 0}^{N - 1} {{\gamma ^i}(1 - \gamma )} {\sigma _0} 
     \nonumber \\
  &= {\gamma ^N}\varepsilon V(0) 
    - \frac{{1 - {\gamma ^N}}}{{1 - \gamma }}(1 - \gamma )\sigma V(0) 
       \nonumber \\
  &= \left( 
        {{\gamma ^N}\varepsilon  - (1 - {\gamma ^N})\sigma } 
     \right)V(0).
\end{align*}
This is positive because we have chosen $\sigma$ as in \eqref{eq-I01}.

Hence, \eqref{eq-18} can be written as
\begin{align}\label{eq-20}
 {x_i}(k + 1)
  & \le {\overline{x} _0}(k) 
     + \left( {1 - \gamma } \right){\sigma _0} 
    - \gamma {\varepsilon _0}(k) 
     + {c_0} + {c_1}{\text{e}^{ - \alpha (k-1)}} 
    \nonumber \\
  & = {\overline{x} _0}(k + 1) 
       - {\varepsilon _0}(k + 1),
\end{align}
where in the inequality, we used the fact that there always 
exists $j$ not in $\overline{\mathcal{X}}_0(k,\varepsilon _0(k))$. 
This relation shows that if an update happens at node $i$, then
this node will move out of $\overline{\mathcal{X}_0}(k + 1,\varepsilon_0(k + 1))$. 
We note that inequality \eqref{eq-20} also holds for the regular nodes 
that are not inside 
$\overline{\mathcal{X}}_0(k,\varepsilon_0(k))$
at time $k$.
This means that such nodes cannot move in 
$\overline{\mathcal{X}}_0(k + 1,\varepsilon_0(k + 1))$. 
It is also similar with $\underline{\mathcal{X}}_0(k + 1,\varepsilon_0(k+1))$.

Thus, after $N$ time steps, all regular nodes will be out of 
at least one of the two sets 
$\overline{\mathcal{X}}_0(N,\varepsilon_0(N))$ 
and $\underline{\mathcal{X}}_0(N,\varepsilon_0(N))$. 
We suppose that $\overline {\mathcal{X}}_0(N,\varepsilon_0(N)) \cap \mathcal{R}$ 
is empty. Then we have $\overline{x}(N) \le \overline{x}_0(N) - \varepsilon_0(N)$.
It hence follows that
\begin{align*}
  &V(N) 
    = \overline{x}(N) - \underline{x}(N) \\
  & \le \overline {x}_0(N) - \varepsilon _0(N) 
        - \underline {x}_0(N) + \sigma_0 \\
  & = \overline{x}_0(0) - \underline{x}_0(0) 
        + 2 c_0 N 
        + 2 \sum\limits_{i=0}^{N - 1} 
              c_1 \text{e}^{-\alpha i}
                - \varepsilon_0(N) + \sigma_0 \\
  & = \left( 
         \overline{x}(0) - \sigma_0
      \right) 
      - \left( 
          \underline{x}(0) + \sigma _0
        \right) + 2 c_0 N 
         + 2 c_1 \frac{1 - \text{e}^{-\alpha N}}{%
                       1 - \text{e}^{-\alpha}}\\
  & \hspace*{5cm}\mbox{}
     - \varepsilon_0(N) + \sigma_0  \\
  & = V(0) 
     + 2{c_0}N 
     + 2{c_1}\frac{1 - \text{e}^{-\alpha N}}{%
                   1 - \text{e}^{-\alpha }}
     -  \sigma V(0) \\
& \hspace*{4cm}\mbox{}
   - \left( 
       \gamma^N \varepsilon  - \left(1 - \gamma ^N\right)\sigma 
     \right)  V( 0 ) \\
& = \left( 
      1 - \gamma ^N (\varepsilon + \sigma)
    \right) V(0) 
     + 2{c_0}N 
     + 2{c_1}\frac{1 - \text{e}^{-\alpha N}}{%
                   1 - \text{e}^{-\alpha }}. 
\end{align*}
By $\varepsilon + \sigma = 1/2$ in \eqref{eq-I01}, we have
\begin{equation*}
  V(N) 
    \le \left( 
           1 - \frac{\gamma ^N}{2} 
         \right) V(0) 
          + 2{c_0}N 
         + 2{c_1}\frac{1 - \text{e}^{-\alpha N}}{%
                       1 - \text{e}^{-\alpha }}. 
\end{equation*}

If there are more updates by node $i$ after time $k=N$, this argument 
can be extended further as
\begin{align}\label{eq-24}
& V( l N) 
  \le \left( 1 - \frac{\gamma ^N}{2} \right) 
          V( (l - 1)  N) 
    \nonumber \\
 & {\kern 80pt} \mbox{}
   + 2{c_0}N 
   + 2{c_1}\frac{{1 - {\text{e}^{ - \alpha N}}}}{%
                 {1 - {\text{e}^{ - \alpha }}}} 
      {\text{e}^{ -(l - 1 )\alpha N}} 
      \nonumber \\
 &\le \left( 
          1 - \frac{\gamma ^N}{2} 
      \right)^l
           V(0)  
    + \sum\limits_{t = 0}^{l - 1} 
         \left( 
           1 - \frac{\gamma ^N}{2} 
         \right)^{l - 1 - t}
    \nonumber \\
 & {\kern 80pt} \mbox{}
   \times \left( 
            2 c_0 N + 2 c_1
          \frac{1 - \text{e}^{ - \alpha N}}{%
                1 - \text{e}^{ - \alpha }}
           \text{e}^{ - (t - 1)\alpha N}
          \right) 
    \nonumber   \\
 & \le \left( 
          1 - \frac{\gamma ^N}{2} 
       \right)^l  V(0) 
     + \frac{1 - \left( 
                   1 - \frac{\gamma ^N}{2}
                 \right)^l}{%
             1 - \left( 
                   1 - \frac{\gamma ^N}{2} 
                 \right)} 2{c_0}N 
    \nonumber \\
& {\kern 15pt} \mbox{}
   + 2{c_1}\frac{1 - \text{e}^{ - \alpha N}}{%
                 1 - \text{e}^{ - \alpha }}
     \left( 
        1 - \frac{\gamma ^N}{2} 
     \right)^l
     \frac{1 - (1 -\frac{\gamma ^N}{2})^{- l} \text{e}^{-\alpha Nl}}{%
           1 - \left( 
                  1 - \frac{\gamma ^N}{2}
               \right)^{ - l} \text{e}^{ - \alpha N}}.
\end{align}
From \eqref{eqn:thm1}, we can easily obtain
\begin{equation}
  \mathop {\limsup }\limits_{l \to \infty } V(lN) 
   \leq \frac{2 c_0 N}{1 - \left( 
                            1 - \frac{\gamma ^N}{2}
                          \right)}
    = \frac{4 c_0 N}{ \gamma^N}
    \leq c. 
\label{eqn:boundV}
\end{equation}

Now, we show the dynamics of $V(lN+t)$ for $t=0,1,\ldots,N-1$. 
The analysis is similar, and we can obtain an inequality 
like \eqref{eq-24}, where the only difference is 
that in the derivation, $V(0)$ is replaced with $V(t)$. 
From the safety condition, we know that $V(k) \leq |\mathcal{S}|$ 
for all $k$. Therefore, we finally arrive at
\[
 \mathop {\limsup }\limits_{l \to \infty } 
   V(l N+t) 
    \le \frac{4 c_0 N}{\gamma ^N}
    \leq c.         
\]
This completes the proof of the consensus condition.
\end{proof}

The above result shows that the multi-agent system 
is guaranteed to reach resilient consensus despite 
the presence of $F$-total malicious agents. 
First, the width of the safety interval $\mathcal{S}$ 
is determined by the initial states of the regular agents. 
Second, the error that may remain after achieving 
resilient consensus meets the specified bound $c$ by
selecting the key parameter in the triggering function
$c_0$, proportionally to $c$. 
This parameter can be set by the designer and, 
clearly, by taking $c_0=0$, exact consensus can 
be achieved at the expense of having more communications. 
The role of $c_1$ and $\alpha$ is to reduce the 
communication during the transient stage by making the 
threshold in the triggering function large. 
We note that the exponential decaying bound 
by $c_1$ and $\alpha$ can also decrease the 
communication in the long run. 
We will see the effects of the parameters of the 
event-triggering function through a numerical example 
in Section~\ref{S:5}. 

In the literature of event-based consensus, 
conventional schemes often employ triggering functions
whose thresholds go to zero over time, in both continuous- 
and discrete-time domains
(e.g., \cite{Dimarogonas2012,Guo2014,Ma2017,MengChen2013,Meng2017}). 
By contrast, \cite{Kia2015,Seyboth2013}
use thresholds which always take positive values as in our study. 
In comparison,
our upper bound for the consensus error is more conservative. 
Because of the malicious agents, the analysis cannot 
apply the methods in previous works and must follow
those in resilient consensus 
problems such as \cite{Dibaji2017}; as a consequence,
the bound on consensus errors grows exponentially
with $N$ (see \eqref{eqn:boundV}). In the
conventional results of \cite{Kia2015,Seyboth2013},
the bounds depend on $N$ linearly as well as on the 
Laplacian matrix.

A related result for the case of $F$-local model 
for the adversarial nodes can be found in 
\cite{Kikuya2018} with a particular application to clock synchronization 
in WSNs. It studies a resilient consensus problem with decaying noise 
that arises in the system due to the interactions among clock states.


\begin{myremark}\label{remark1}\rm
We should highlight that in the discrete-time domain,
event-based consensus algorithms must be carefully designed
especially in the resilient case. 
We can construct another resilient consensus algorithm 
motivated by the structures found in \cite{Seyboth2013,Xiao2008},
which deal with continuous-time multi-agent systems, as 
\begin{equation} \label{eq-**03}
 {x_i}(k+1)
   = {x_i}(k)
      + \sum\limits_{j \in {\mathcal{M}_i}(k)} 
         a_{ij}(k) 
         \left( 
            \hat{x}_j(k) - \hat{x}_i(k)
         \right).
\end{equation}
The modification may be minor as the only difference
is that $\hat{x}_i(k)$ is used instead of $x_i(k)$ in 
the second term of the right-hand side.
Compared with Protocol~1, to guarantee the consensus 
error level of $c$, the choice of $c_0$ must be
half as $c_0\leq \gamma^N/8N$, which may increase
the communication load. These results can be obtained 
by following a proof similar to that of Theorem~\ref{theorem1}.

\end{myremark}

In the next section, we present yet another protocol by further
changing the terms in the update rule. 

\section{Protocol~2}
\label{S:4}

In this section, we provide our second resilient consensus algorithm,
referred to as Protocol~2. 

To this end, we modify the update rule \eqref{eq-03}
in a way different from the protocol \eqref{eq-**03} discussed 
in Remark~\ref{remark1}.
It is pointed out that in Protocol~1, 
for obtaining the new state $x_i(k+1)$ of agent~$i$, 
its own data appears only through the current state $x_i(k)$.
On the one hand, this means that even when the new state is not communicated, 
it still needs to be stored at every time step. 
On the other, as the current state $x_i(k)$ is newer than $\hat{x}_i(k)$, 
it seems desirable for speeding up the convergence.  
We will however show that it may be better to use only $\hat{x}_i(k)$
for both storage and convergence reasons.
The protocol introduced below is motivated 
by those in \cite{Kadowaki2015,Xiao2008}.

In the local update, for $k\in \mathbb{Z_+}$, every regular node $i\in \mathcal{R}$ updates its current state by
\begin{equation}\label{eq-27}
  x_i(k+1)
    = \hat x_i(k) 
       + \sum\limits_{j \in {\mathcal{M}_i}(k)}
           a_{ij}(k) 
       \left( 
           \hat{x}_j(k) - \hat{x}_i(k)
       \right).
\end{equation}
Note that the new state $x_i(k+1)$ need not be stored until 
the next time step,
but is merely used for checking the condition of the 
triggering function $f_i(k)$ in \eqref{eq-02}. 
Accordingly, in the E-MSR, steps~1 and~2 should be adjusted so that
agent~$i$ uses $\hat{x}_i(k)$ instead of $x_i(k)$ in determining
the neighbor set $\mathcal{M}_i(k)$.

Then we are ready to present our second main result of this paper, 
which is regarding Protocol~2.

\begin{theorem} \label{theorem2}\rm
Under the $F$-total malicious model, the normal agents with E-MSR using \eqref{eq-27} and \eqref{eq-04} reach resilient consensus if and only if the underlying graph is $(F+1,F+1)$-robust. The safety interval is given by 
$\mathcal{S} = \left[ {\underline {\hat{x}} (0),\overline {\hat{x}} (0)} \right]$, 
and the consensus error level $c$ is achieved if
the parameter $c_0$ in the triggering function \eqref{eq-02}
satisfies
\begin{equation}\label{th2}
  c_0 \leq \frac{\gamma^{N - 1}( 1 - \gamma )c}{%
                     1 - \gamma^{N-1}}.
\end{equation}
\end{theorem}

\begin{proof}
The necessity part follows similar lines as those in the proof of 
Theorem~\ref{theorem1}. In the following, we thus give the sufficiency part. 

First, we establish the safety condition in the sense of
$x_i(k), \hat{x}_i(k)\in\mathcal{S}$ for regular nodes $i$.
This is done by induction.
At $k=0$, for each $i\in\mathcal{R}$, it holds
$x_i(0), \hat{x}_i(0)\in\mathcal{S}$ by definition.
Next, assume that at time $k$, we have 
$x_i(k), \hat{x}_i(k)\in \mathcal{S}$ for $i\in\mathcal{R}$.
Then, for agent~$i$, its neighbors $j\in\mathcal{M}_i(k) $ 
satisfy $\hat{x}_j(k)\in\mathcal{S}$ since there are at most $F$
agents with values outside $\mathcal{S}$, and they are
ignored in step~2 of the E-MSR. 
From the update rule \eqref{eq-27}, we have
\begin{align} \label{eq-l01}
 x_i(k+1)
  & = a_{ii}(k) \hat{x}_i(k) 
      + \sum\limits_{j \in \mathcal{M}_i(k)} a_{ij}(k) \hat{x}_j(k) 
      \nonumber \\
  & \le  a_{ii}(k) \overline{\hat{x}}(k) 
       + \sum\limits_{j \in \mathcal{M}_i(k)}
        a_{ij}(k) \overline{\hat{x}} (k)
   = \overline{\hat{x}}(k),
\end{align}
where $a_{ii}(k) 
   = 1 - \sum_{j \in {\mathcal{M}_i}(k)} 
           a_{ij}(k)$.
The inequality \eqref{eq-l01} means that the upper 
bound of every regular node is nonincreasing. 
Similarly, we have 
$x_i(k+1)\geq \underline {\hat x}(k)$, 
so we obtain 
$x_i(k)\in \mathcal{S}$
for $k \ge 0$.
Furthermore, by \eqref{eq-04}, it holds that
$\hat{x}_i(k + 1) \in \mathcal{S}$. 
Hence, we have $\mathcal{S}$ as the safety interval.

For the consensus condition part, we first sort the regular 
communicated values ${\hat x}_i(k), i \in \mathcal{R}$, 
at time $k$ in the entire graph from the smallest to the 
largest. Denote by $s_i(k)$ 
the index of the agent taking the $i$th value 
from the smallest.
Hence, the values are sorted as 
$\hat{x}_{s_1}(k) \le \hat{x}_{s_2}(k) \le \cdots \le \hat{x}_{s_{N}}(k)$.

Introduce two sequences of conditions for the relation
of each gap between two nodes. The first is given from below as
\begin{itemize}
  \item $A_1$:  $\hat{x}_{s_2}(k) - \hat{x}_{s_1}(k) 
                    \le (c_0+c_1\text{e}^{-\alpha k})/{\gamma}$,
  \item $A_2$:  $\hat{x}_{s_3}(k) - \hat{x}_{s_2}(k) 
                    \le (c_0+c_1 \text{e}^{-\alpha k})/{\gamma^{2}}$,
  \item $\cdots$
  \item $A_{N-1}$:  $\hat{x}_{s_N}(k) - \hat{x}_{s_{N-1}}(k) 
                    \le (c_0+c_1\text{e}^{-\alpha k})/{\gamma^{N-1}}$.
\end{itemize}
The other sequence is from above as
\begin{itemize}
  \item $B_{N}$:  $\hat{x}_{s_N}(k) - \hat{x}_{s_{N-1}}(k) 
                   \le (c_0+c_1\text{e}^{-\alpha k})/{\gamma}$,
  \item $B_{N-1}$:  $\hat{x}_{s_{N-1}}(k) - \hat{x}_{s_{N-2}}(k) 
                   \le (c_0+c_1\text{e}^{-\alpha k})/{\gamma^{2}}$,
  \item $\cdots$
  \item $B_2$:  $\hat{x}_{s_2}(k) - \hat{x}_{s_1}(k) 
                   \le (c_0 + c_1\text{e}^{-\alpha k})/{\gamma^{N-1}}$.
\end{itemize}
Let $j_A$ be the minimum $j=1,\ldots,N-1$ such that 
condition $A_{j}$ is not satisfied. Also, let $j_B$ 
be the maximum $j=2,\ldots,N$ such that condition $B_{j}$ 
is not satisfied. Thus we have
\begin{equation} 
\begin{split}
  & \hat{x}_{s_{j_A+1}}(k) - \hat{x}_{s_{j_A}}(k) 
       > \frac {c_0 + c_1 \text{e}^{-\alpha k} }{\gamma^{j_A}}, \\
  & \hat{x}_{s_{j_B}}(k) - \hat{x}_{s_{j_B-1}}(k) 
       > \frac {c_0 + c_1 \text{e}^{-\alpha k} }{\gamma^{N-j_B+1}}.
\end{split}
\label{eq-**04}
\end{equation}
Moreover, conditions $A_1$ to $A_{j_A-1}$ and $B_{j_B+1}$ 
to $B_{N}$ are satisfied. Then,
for $0 \le k \le k'$, we introduce two sets
\begin{align*}
  \mathcal{X}_{1}(k, k') 
    &= \left\{ 
        j \in \mathcal{V}:~
         \hat{x}_j(k') 
           < \hat{x}_{s_{j_A}}(k) 
               + c_0 + c_1\text{e}^{-\alpha k} 
        \right\},\\
  \mathcal{X}_{2}(k, k') 
    &= \left\{ 
           j \in \mathcal{V}:~
             \hat{x}_j(k') > \hat{x}_{s_{j_B}}(k) 
               - c_0 - c_1\text{e}^{-\alpha k} 
         \right\}.
\end{align*}

There are two cases concerning the relationship between 
$j_A$ and $j_B$. We study them separately below.

Case 1: $j_A<j_B$. Let the two subsets of the regular nodes be
$\mathcal{V}_1=\left\{s_1(k), s_2(k),\ldots, s_{j_A}(k)\right\}$
and $\mathcal{V}_2=\left\{s_{j_B}(k),\ldots,s_{N}(k)\right\}$. 
Note that all nodes in $\mathcal{V}_1$ are inside 
$\mathcal{X}_{1}(k, k)$, and those in $\mathcal{V}_2$ 
are inside $\mathcal{X}_{2}(k, k)$. 
Hence, $\mathcal{X}_{1}(k, k)$ and $\mathcal{X}_{2}(k, k)$ are nonempty.
They are moreover disjoint. This is because by using the
two inequalities in \eqref{eq-**04},
from $1\leq j_A < j_B\leq N$ and $0<\gamma\leq 1/2$, 
it follows that 
\begin{align*}
\hat{x}_{s_{j_B}}(k) - \hat{x}_{s_{j_A}}(k) 
  &> \max \left\{
            \frac{1}{\gamma^{j_A}}, \frac{1}{\gamma^{N-j_B+1}}
         \right\}
          \left(
             c_0 + c_1 \text{e}^{-\alpha k} 
          \right)\\
  &\geq 2 \left(
             c_0 + c_1 \text{e}^{-\alpha k} 
          \right).
\end{align*}
Thus, the $(F+1,F+1)$-robust graph guarantees that some
regular node $i$ in 
$\mathcal{X}_{1}(k, k)$ or $\mathcal{X}_{2}(k, k)$ 
has at least $F+1$ neighbors outside. 
We suppose that $i \in \mathcal{X}_1(k,k)$. 
By \eqref{eq-27}, 
\begin{align*}
 {x_i}(k + 1)
   &= a_{ii}(k) \hat{x}_i(k)
       + \sum\limits_{j \in {\mathcal{M}_i}(k) \cap \mathcal{X}_1} 
          a_{ij}(k)\hat{x}_j(k)\\
   &\hspace*{3cm}\mbox{}
       + \sum\limits_{j \in {\mathcal{M}_i}(k) \setminus \mathcal{X}_1} 
          a_{ij}(k) \hat{x}_j(k),
\end{align*}
where the simplified notation $\mathcal{X}_1$ is used for $\mathcal{X}_1(k,k)$.
Since  ${\mathcal{M}_i}(k) \setminus \mathcal{X}_1(k,k)$ is 
not empty, we have
\begin{align} \label{eq-l05}
  x_i(k + 1)
   & \ge ( 1 - \gamma ) \hat{x}_{s_1}(k) 
      + \gamma \hat{x}_{s_{j_A+1}}(k).
\end{align}
Using conditions $A_1$ to $A_{j_A-1}$, 
we can bound $\hat{x}_{s_1}(k)$ 
as
\begin{align*} 
 &\hat{x}_{s_1}(k)
   \ge \hat{x}_{s_2}(k) 
         - \frac{c_0 + c_1\text{e}^{-\alpha k}}{\gamma} \\
  &~~ \ge \hat{x}_{s_3}(k) 
         - \left( 
             \frac{1}{\gamma} 
                + \frac{1}{\gamma^{2}} 
            \right)
          \left( c_0 + c_1\text{e}^{-\alpha k} \right) 
      \ge \cdots  \\
  &~~ \ge \hat{x}_{s_{j_A}}(k) 
         - \left( 
             \frac{1}{\gamma} 
                + \frac{1}{\gamma^{2}} 
                + \cdots + \frac{1}{\gamma^{j_A-1}} 
           \right)
             \left( 
               c_0 + c_1\text{e}^{-\alpha k} 
             \right).
\end{align*}
Substitute this into \eqref{eq-l05} and obtain
\begin{align} \label{eq-l07}
  x_i(k + 1)
   & \ge \hat{x}_{s_{j_A}}(k) 
       + \gamma \left( 
                 \hat{x}_{s_{j_A + 1}}(k) - \hat{x}_{s_{j_A}}(k) 
                \right) \nonumber \\
  &  {\kern 15pt} \mbox{} 
      - \frac{1}{\gamma ^{j_A - 1}}{(c_0+c_1\text{e}^{-\alpha k})} 
      + (c_0 + c_1\text{e}^{-\alpha k})
         \nonumber \\
  & > \hat{x}_{s_{j_A}}(k) 
        + \gamma \frac{c_0+c_1\text{e}^{-\alpha k}}{\gamma ^{j_A}} 
        \nonumber \\
  & {\kern 15pt} \mbox{} 
     - \frac{1}{\gamma ^{j_A - 1}}{(c_0+c_1\text{e}^{-\alpha k})} 
     + {(c_0+c_1\text{e}^{-\alpha k})} 
        \nonumber \\
  & = \hat{x}_{s_{j_A}}(k) + {(c_0+c_1\text{e}^{-\alpha k})},
\end{align}
where the second inequality follows by \eqref{eq-**04}. 
Thus, this node $i$ is moved out of set 
$\mathcal{X}_{1}(k, k+1)$ 
at time $k+1$.

We next show that the regular nodes not in $\mathcal{X}_1(k,k)$ at time $k$
will not move in $\mathcal{X}_{1}(k, k+1)$ at time $k+1$. 
If node $j$ has some neighbors inside $\mathcal{X}_1(k,k)$, 
then \eqref{eq-l05} and \eqref{eq-l07} hold and 
we know that the node does not move in $\mathcal{X}_{1}(k, k+1)$. 
If node $j$ has neighbors only in 
$\mathcal{V} \setminus \mathcal{X}_1(k,k)$, then we have
\[
   x_j(k+1) 
     \ge \hat{x}_{s_{j_A+1}}(k) 
     > \hat{x}_{s_{j_A}}(k) 
        + \frac{c_0 + c_1 \text{e}^{-\alpha k}}{\gamma^{j_A}}.
\]
Clearly, node $j$ does not move in $\mathcal{X}_{1}(k, k+1)$ 
in this case. 

Therefore, the regular nodes in 
$\mathcal{X}_{1}(k,k+1)$ decrease in number as
$\mathcal{X}_{1}(k, k+1) \cap \mathcal{R} 
\subsetneq {\underline { \mathcal{X}}_{1}}(k, k) \cap \mathcal{R}$.
Similar results also hold if $i \in \mathcal{X}_2(k,k)$, and we have 
${\hat x}_i(k+1)$ decreases more than $c_0+c_1\text{e}^{-\alpha k}$ 
compared with $\hat{x}_{s_{j_B}}(k)$.

As a result, if conditions $A_{j_A}$ and $B_{j_B}$ with $j_A< j_B$ are 
not satisfied, after $N$ steps, the set 
$\mathcal{X}_{1}(k, k+N)$ or 
$\mathcal{X}_{2}(k, k+N)$ becomes empty 
in regular nodes. It then follows that $\underline {\hat x}(k+N)$ 
increases more than $c_0+c_1\text{e}^{-\alpha k}$ or 
$\overline {\hat x}(k+N)$ decreases more than 
$c_0+c_1\text{e}^{-\alpha k}$.

A special case in Case 1 is when $j_A=j_B-1$. 
It corresponds to having
only one pair of nodes whose difference 
in values does not satisfy the condition. 
By applying a similar analysis, we have that
$\underline {\hat x}(k+N)$ increases more than 
$c_0+c_1\text{e}^{-\alpha k}$ or $\overline {\hat x}(k+N)$ 
decreases more than $c_0+c_1\text{e}^{-\alpha N}$.

Case 2: $j_A \ge j_B$. This case is impossible. We can 
show this by contradiction as follows. Since $j_A \ge j_B$, 
we know that $A_{j_B-1}$ and $B_{j_A+1}$ are both satisfied. 
Combined with $A_{j_A}$ and $B_{j_B}$ not being satisfied, 
we have
\begin{align}
 \frac {{c_0+c_1\text{e}^{-\alpha k}}}{\gamma^{N-j_B+1}}  
   &< \hat{x}_{s_{j_B}}(k) - \hat{x}_{s_{j_B-1}}(k)
   \le \frac {{c_0+c_1\text{e}^{-\alpha k}}}{\gamma^{j_B-1}},
   \label{ssA1}\\
 \frac{{c_0+c_1\text{e}^{-\alpha k}}}{\gamma^{j_A}} 
   &< \hat{x}_{s_{j_A+1}}(k) - \hat{x}_{s_{j_A}}(k)
   \le \frac {{c_0+c_1\text{e}^{-\alpha k}}}{\gamma^{N-j_A}}.
   \label{ssA2}
\end{align}
The inequalities in the first relations in \eqref{ssA1} 
indicate that it must hold 
$j_B > (N+1)/2$. The second set of inequalities in \eqref{ssA2} 
also implies $j_A < N/2$. Consequently, we have $j_A < j_B$, which is in 
contradiction with $j_A \ge j_B$.

We can now conclude that after a finite number of time steps, 
all conditions from $A_1$ to $A_m$ and $B_{m+2}$ to $B_{N}$, 
where $0 \le m \le N-1$, must be satisfied. 
Otherwise the difference between $\overline {\hat x}(k)$ 
and $\underline {\hat x}(k)$ will decrease more than $c_0$ by an update induced by an event. From the analysis for the safety condition, 
we know that $\overline {\hat x}(k)$ is nonincreasing 
and $\underline {\hat x}(k)$ is nondecreasing. 
Hence, if the events continuously occur, 
$\overline {\hat x}(k) - \underline {\hat x}(k)$ will 
become smaller and eventually negative, which cannot happen. 
This completes the proof.
\end{proof}

Protocol~2 enables us to achieve resilient consensus 
with data communicated via event-based protocols. 
We emphasize that our analysis for Protocol~2 is less 
conservative compared to Theorem~\ref{theorem1} for Protocol~1.
In fact, we can see by directly comparing 
the two theorems that the bound on the parameter $c_0$ is larger
for achieving the same level $c$ of consensus error;
this may result in less frequent transmissions. 
We will confirm this property later in Section~\ref{S:5} 
through numerical simulations.

A unique aspect of Protocol~2 is that
the proof technique used in Theorem~\ref{theorem2} 
is different 
from those used in the recent works such as 
\cite{Dibaji2017,Dibaji2018,LeBlanc2017,LeBlanc2013}
and also in the proof of Theorem~\ref{theorem1}. 
The conventional technique could be employed here, but this will 
result in the same bound on $c_0$ as in Theorem~\ref{theorem1}.
In fact, as we see below, the bound obtained in Theorem~\ref{theorem2} 
is tight for some graphs.

\begin{myremark}\rm
We present an example of a multi-agent system whose error 
in consensus among the agents is equal to the bound obtained 
in Theorem~\ref{theorem2}. Such a graph may be called 
a worst-case graph. Consider the network in Fig.~\ref{fig1} 
with four nodes which are all regular and thus $F=0$. 
Note that the graph contains a directed spanning tree.
The initial values $x_i(0)$ of the nodes and the (constant) 
weights $a_{ij}(k)$ on the edges are indicated in the figure. 
Since the weights are all $1/2$ (and thus $\gamma=1/2$), 
for nodes having two neighbors, their own values are not 
used in the update rule \eqref{eq-27}. 
Moreover, for the node in the far left, a self-loop is 
shown to indicate that this node uses its own value. 
The node in the far right has no incoming edge, and 
thus its value will not change over time. 

By setting the parameters for the triggering function 
as $c_0=1$ and $c_1=0$, it follows that there will be 
no event at any time. The difference in their 
values is 14, which can be obtained as $c$ by equating
the inequality~\eqref{th2} in Theorem~\ref{theorem2}. 
In comparison, for Protocol~1, the bound $c$ on the difference
will be 256 by Theorem~\ref{theorem1}; this is much larger,
indicating the conservatism of the analysis approach there.
Note that the graph structure in Fig.~\ref{fig1} is obtained based on
the proof of Theorem~\ref{theorem2}. It is a bit special in the sense
that not all agents have self-loops. To comply with the theory, 
we can extend this example by adding self-loops; it will not be
a worst-case graph any longer, but the difference among the
values will be larger than other graphs. 
\end{myremark}

\begin{figure}[t]
\centering
\includegraphics[width=0.62\linewidth]{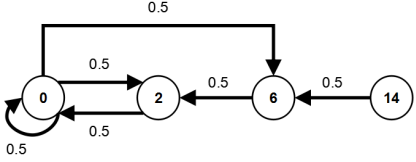}
\vspace*{-2mm}
\caption{Worst-case graph with $N$=4}
\label{fig1}
\vspace*{-4mm}
\end{figure}

\section{Numerical Example}
\label{S:5}

In this section, we illustrate the proposed 
resilient protocols via numerical 
simulations. We first examine a small-scale 
network and then focus on the scalability for 
larger systems. 

\vspace*{-2mm}
\subsection{Small network}

We consider the multi-agent system with seven 
nodes whose connectivity graph is shown in 
Fig.~\ref{fig2}; as already mentioned, 
this graph is $(3,3)$-robust. 
We compare the performance of 
Protocols~1 and~2
using different parameters in event-triggering. 
In particular, we test the two cases of $c_0>0$ 
and $c_0=0$.
Here, nodes~5 and~7 are set to behave maliciously
by continuously oscillating their values; 
in all simulations, we used the same state values for them.
The initial state was chosen the same for each run as well at 
$x(0)=[1 {\kern 5pt} 2 {\kern 5pt} 3 {\kern 5pt} 
5 {\kern 5pt} 4 {\kern 5pt} 6 {\kern 5pt} 4]^{T}$.
We also took $\gamma=0.3$. 

First, we examine the case of $c_0>0$. 
We fixed the consensus error bound as $c = 1$. 
For Protocol~1, based on Theorem~\ref{theorem1}, 
we chose $c_0 = 1.22 \times 10^{-4}$.
The remaining parameters were selected as
$c_1 = 0.5$ and $\alpha =0.03$. 
The time responses are shown in Fig.~\ref{fig3},
where the $x$-axis represents the 
sampling time $k$, and the $y$-axis 
the values of the agents. Moreover, the time 
instants when each node makes a broadcast are 
shown by the markers $\bullet$ in the color 
corresponding to that of its time response curve.
On the other hand, for Protocol 2, 
we chose $c_0 = 5.72 \times 10^{-3}$
according to Theorem~\ref{theorem2}, and
other parameters were taken as 
above with $c_1 = 0.5$ and 
$\alpha =0.03$. 
The time responses of Protocol~2 are plotted 
in Fig.~\ref{fig4}.  

We observe that both protocols managed to 
achieve the desired level of consensus 
specified by $c=1$ based on event-triggered communication.
Moreover, there is very little 
sign of being influenced by the behavior 
of the malicious nodes.
In fact, for Protocol~1, after 600 steps, 
the consensus error among the regular nodes became 
$5.24 \times 10^{-5}$, with 5.4 times of transmissions
on average for the regular nodes.
On the other hand, for Protocol~2, the consensus
error was $8.63 \times 10^{-3}$, with 4.6 times
of transmissions on average, which is slightly smaller. 
Thus, we confirm that Protocol~2 is
less conservative.

\begin{figure}[t]
\centering
\includegraphics[width=0.94\linewidth]{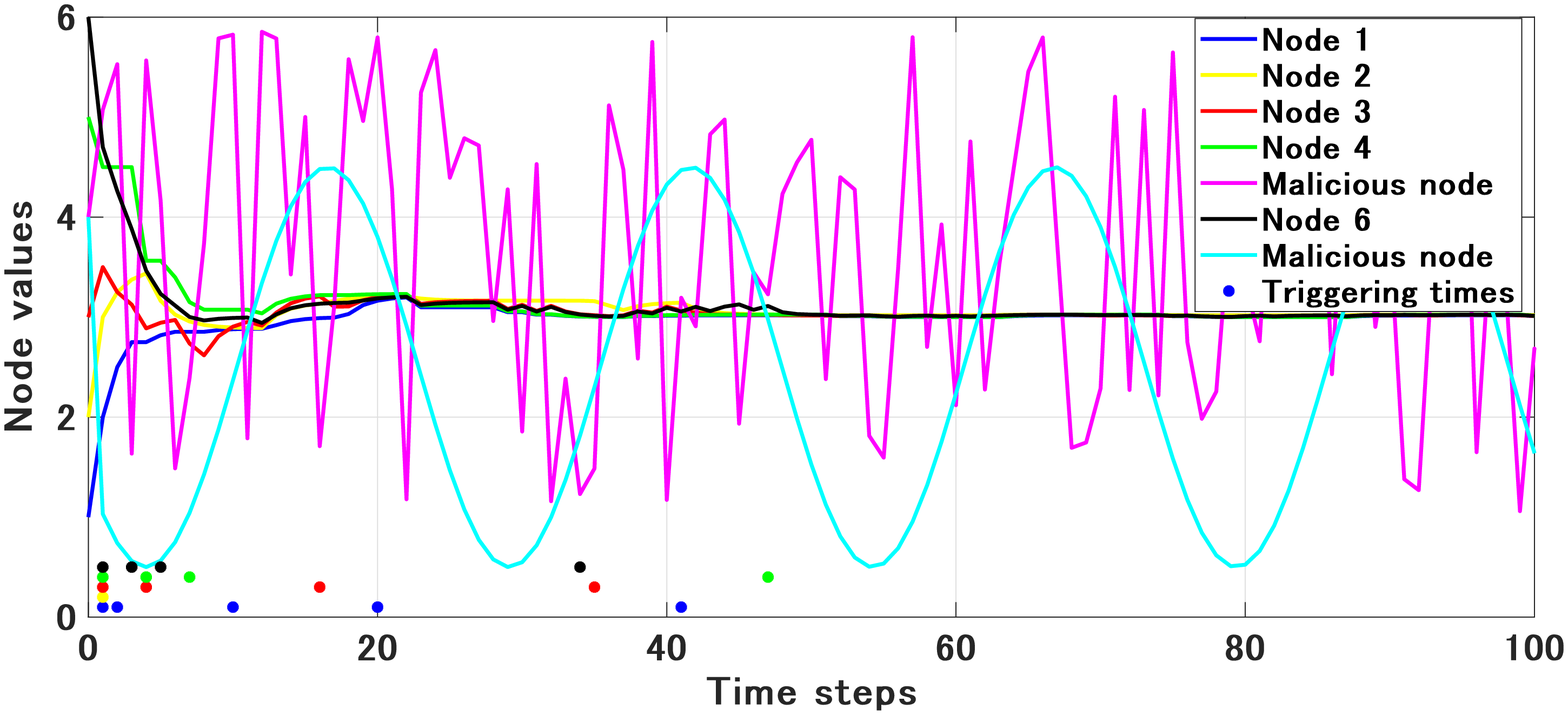}
\vspace*{-2.5mm}
\caption{Protocol~1 with 
$c_0 = 1.215 \times 10^{-4}$, $c_1 = 0.5$, 
and $\alpha = 0.03$}
\label{fig3}
\vspace*{2mm}
\includegraphics[width=0.94\linewidth]{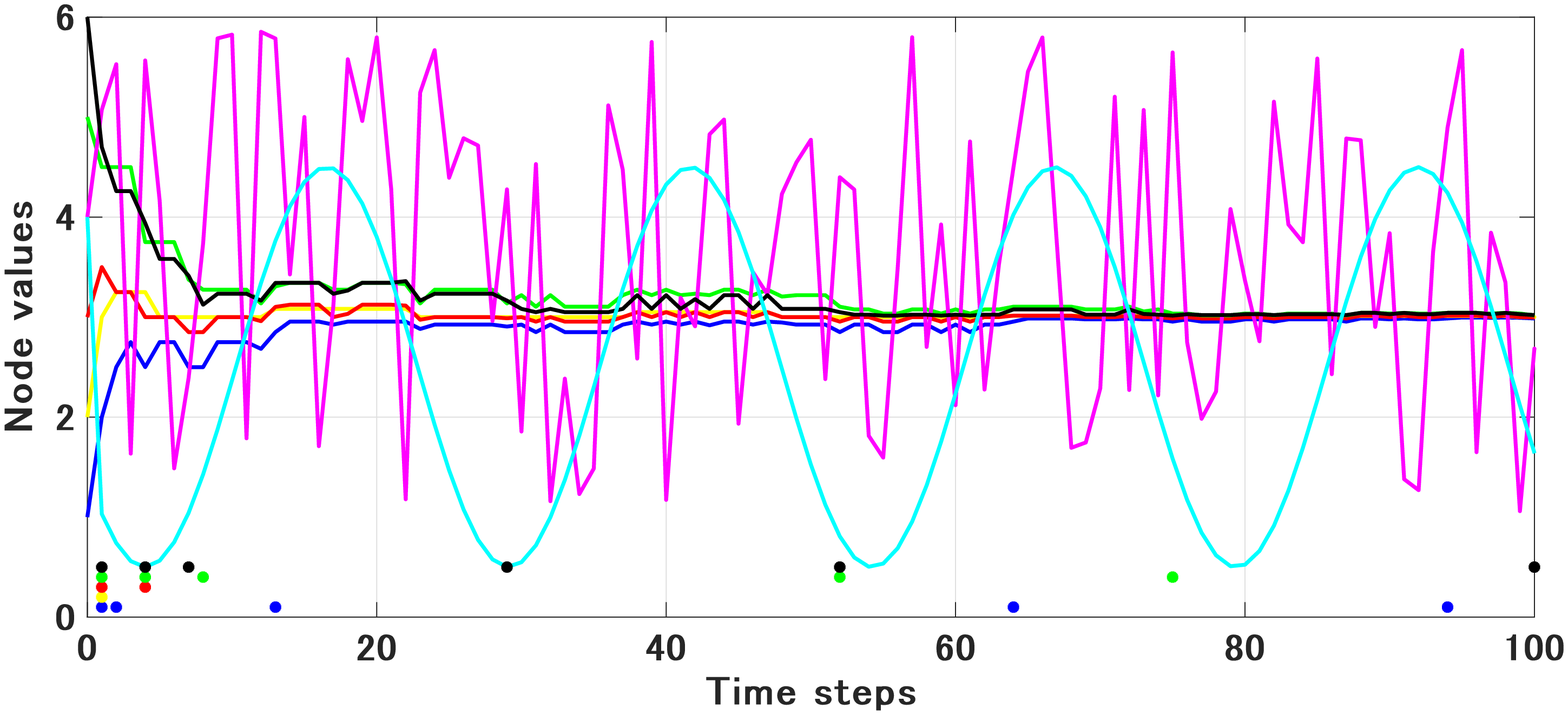}
\vspace*{-2.5mm}
\caption{Protocol~2 with
$c_0 = 5.72 \times 10^{-3}$, $c_1 = 0.5$, 
and $\alpha = 0.03$}
\label{fig4}
\vspace*{-4mm}
\end{figure}

Next, by setting $c_0=0$, 
we demonstrate that exact resilient consensus 
can be attained while
reducing the number of transmissions. 
To this end, for both protocols, we set
$c_1 = 0.5$ and $\alpha = 0.03$ as in the previous
simulations. In this case, the threshold that
determines the timings of events
eventually goes to zero (due to $c_0=0$).
The time responses of the two protocols are 
shown in Figs.~\ref{fig5} and \ref{fig6}.  
For Protocol~1, after 600 steps, 
the consensus errors among the regular nodes 
became essentially zero at $5.71 \times 10^{-9}$,
where the average number of
triggering times for the regular nodes is 10.
Similarly, for Protocol~2, the consensus error 
at time $k=600$ was $1.73 \times 10^{-8}$ with 
12.4 triggering times on average per regular node. 
Protocol~1 is particularly impressive
in terms of limited communication.
In contrast, for Protocol~2, information exchange
among nodes takes place for a longer time.

Further comparisons were made by implementing 
time-triggering communication in both protocols.
Periodic transmissions are made so that after
600 time steps, the regular nodes make the
same number of triggering times as those in the
event-triggered case with $c_0=0$ above. 
This means that for Protocol~1, each node transmits 
every 60 steps and for Protocol~2 every 50 steps. 
At time $k=600$, the consensus error was 
$5.04 \times 10^{-8}$ for Protocol~1 and
$5.80 \times 10^{-3}$ for Protocol~2.
It is clear that under both protocols, 
the event-triggered schemes perform better.
Their time responses are shown in Figs.~\ref{fig7} 
and~\ref{fig8}. Due to the periodic transmission, 
the convergence is slow and the responses between 
the transmission times are oscillatory.


\vspace*{-2mm}
\subsection{Scalability of the proposed approach}

In this part, we carry out a number of simulations 
to check the scalability of the proposed protocols 
using large-scale networks. In particular, we focus
on how the number of transmissions can be kept
low even if the numbers of neighbors and 
the malicious ones are large. 
As in the previous simulations, 
the two cases of $c_0>0$ and $c_0=0$ are examined
and compared with the time-triggered case.
We employ three complete graphs with 10, 
50, and 100 nodes. By Lemma~\ref{lemma1}, 
we know that a 10-node complete graph is 
$(5,5)$-robust. Thus, we introduce four malicious 
nodes. Similarly, in the 50- and 100-node cases,  
we set 24 and 49 malicious nodes, respectively.

\begin{figure}[t]
\centering
\includegraphics[width=0.94\linewidth]{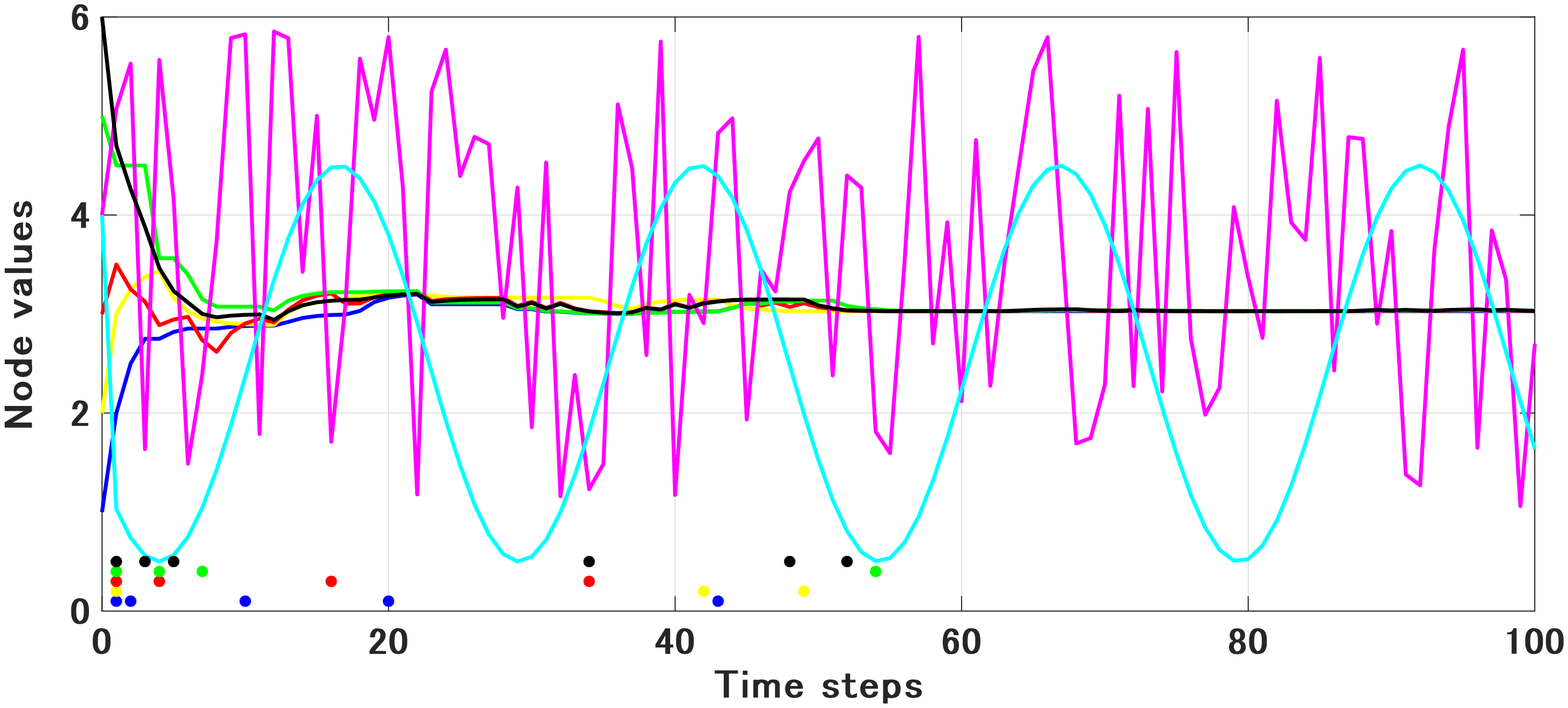}
\vspace*{-2mm}
\caption{Protocol~1 with $c_0 =0$, 
$c_1 = 0.5$, and $\alpha = 0.03$}
\label{fig5}
\vspace*{2mm}
\includegraphics[width=0.94\linewidth]{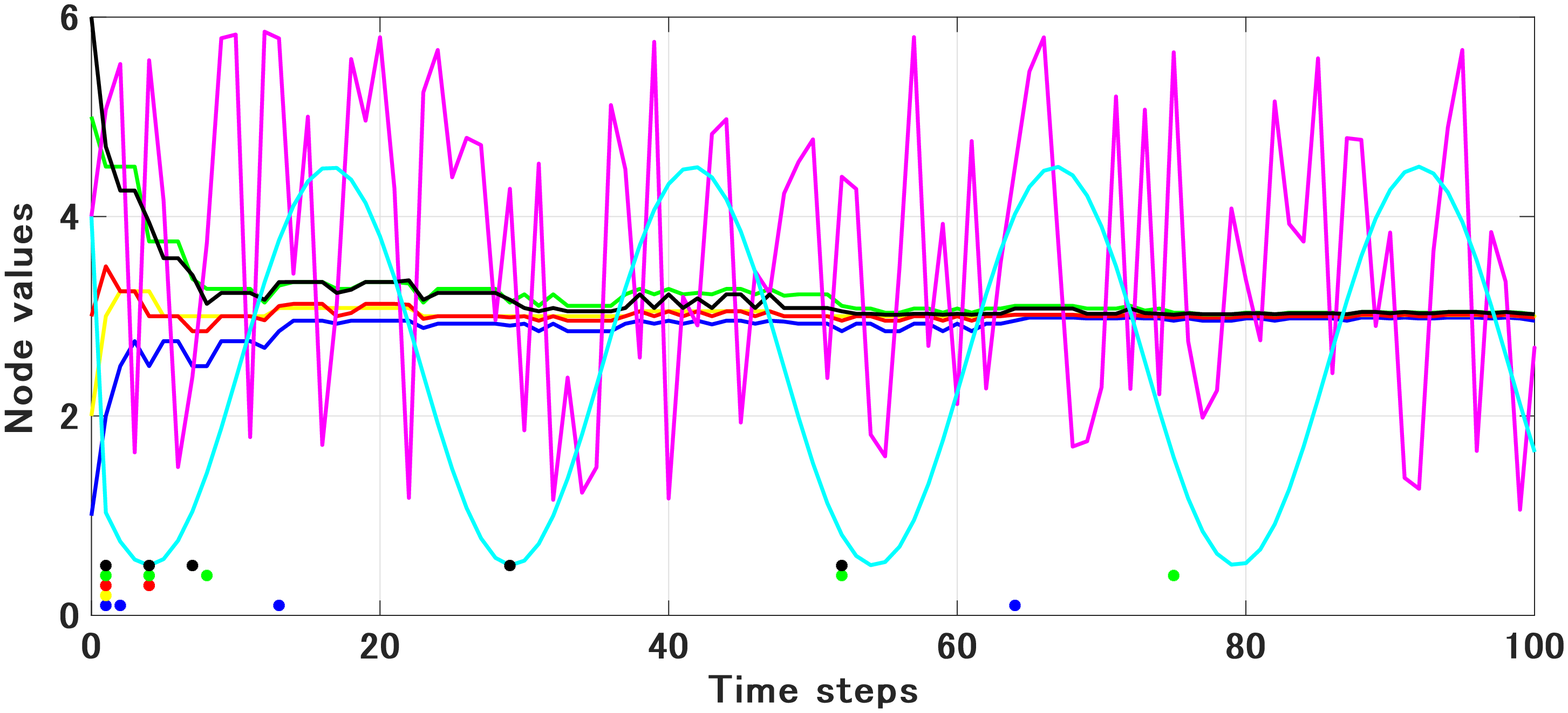}
\vspace*{-2mm}
\caption{Protocol~2 with 
$c_0 =0$, $c_1 = 0.5$, and $\alpha = 0.03$}
\label{fig6}
\vspace*{-4mm}
\end{figure}

The first case is with $c_0 > 0$. In particular, 
for both Protocols~1 and~2, 
we chose $c_0=0.1$, $c_1 = 1$, and $\alpha =2$.  
For each graph, we performed Monte Carlo simulations 
for 100 runs by randomly taking initial states under uniform
distribution between 0 and 100. 
Each agent made updates until the consensus error becomes
0.01 for Protocol~1 and 0.3 for Protocol~2.
The malicious agents made oscillatory behaviors as in
the simulations in the previous subsection.
The performance of Protocols~1 and~2 is
displayed in Tables~\ref{tab1}\,(a) and~(b), respectively,
in terms of
the average number of triggering times per regular node.

It is noticed that in general, as the number of
agents increases, triggering times increase only 
mildly to reach the same level of consensus error
for both protocols.
There is a difference in the achievable performance 
between the protocols as discussed after 
Theorem~\ref{theorem2}. In particular, for the
same size of $c_0$, Protocol~1 is able to yield 
smaller error than Protocol~2.

%
\begin{figure}[t]
\centering
\includegraphics[width=0.94\linewidth]{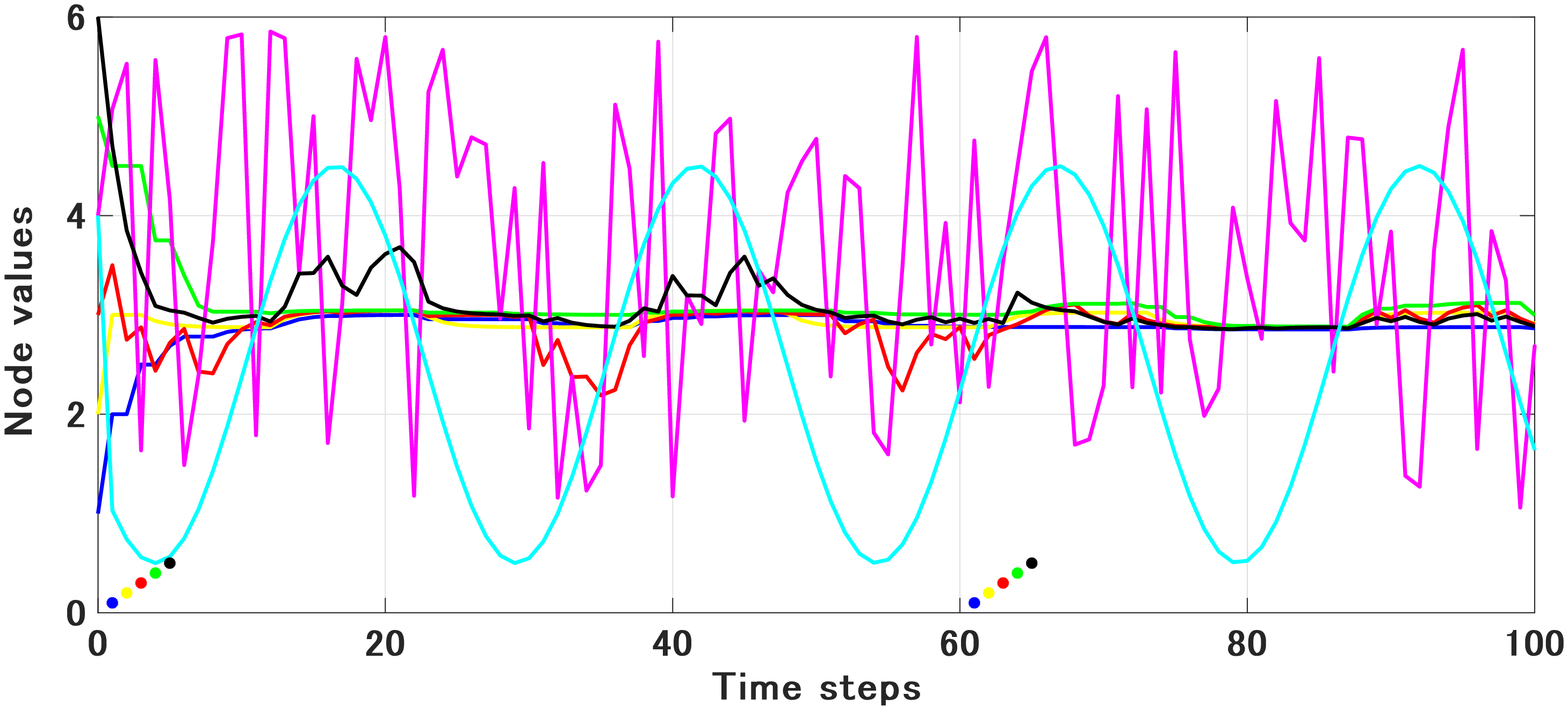}
\vspace*{-2mm}
\caption{Protocol~1 under periodic communication
with period 60}
\label{fig7}
\vspace*{2mm}
\includegraphics[width=0.94\linewidth]{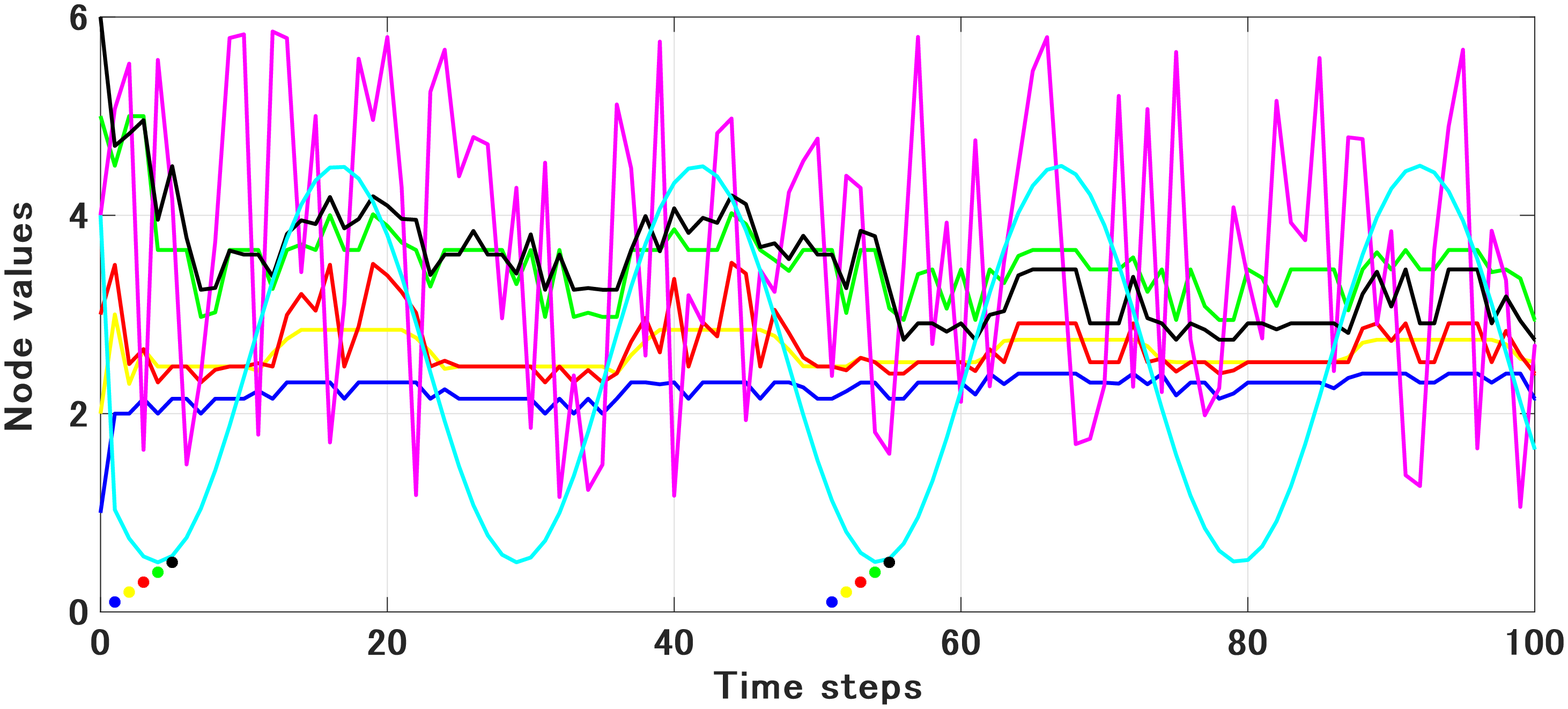}
\vspace*{-2mm}
\caption{Protocol~2 under periodic communication
with period 50}
\label{fig8}
\vspace*{-4mm}
\end{figure}

We proceed to the second case with $c_0=0$.
Other parameters were set as $c_1 = 0.5$ for both
protocols, and as $\alpha$, we used 
0.05 for Protocol~1 and 0.05 for Protocol~2.
The results are shown in the same tables.
The numbers of triggering times are similar
to the case with $c_0>0$.
For Protocol~2, we may say that the scalability 
is slightly less since as the graph sizes increase,
the triggering times increase more.

Finally, in the two tables, we display the
average number of triggering times for the 
time-triggered case, where every node transmits at 
every time step. 
It is evident that such a protocol requires more transmissions
than event-triggered protocols.
From these results, we conclude that 
the event-based protocols can efficiently 
eliminate the amount of communications.



\section{Conclusion}
\label{S:6}

In this paper, we considered a resilient approach for 
the multi-agent consensus problem to mitigate the influence 
of misbehaving agents due to faults and cyber-attacks. 
Two protocols for the updates of the regular nodes have 
been proposed, and their convergence properties as well 
as necessary network structures have been characterized. 
In both cases, resilient consensus can be achieved with 
reduced frequencies in communication among agents through 
event triggering. This is possible at the expense of 
certain errors in consensus determined by the parameters 
in the triggering function. 

Future studies will focus on 
resilient consensus algorithms with 
time delays in communications for the event-triggered case
and also those based on model predictive control. 
A more challenging problem for future research is to construct
algorithms enabling the regular nodes to reach 
a consensus value which is determined only by their initial 
values and not influenced by the adversaries. 


\begin{table}[t]
\centering
\caption{%
Average number of triggering times per regular node}
\vspace*{-3mm}
(a)~Protocol 1 with consensus error 0.01\\
\vspace*{1mm}
\begin{tabular}{|c||c|c|c|}
\hline
       & \multicolumn{2}{|c|}{Event-Triggered} 
       & Time- \\ \cline{2-3}
Graphs & $c_0=0.1$ & $c_0=0$ & Triggered \\ \hline
10 nodes  & 4.9 & 4.4 & 9.8 \\ \hline
50 nodes  & 6.5 & 5.4 & 11.4 \\ \hline
100 nodes & 7.1 & 5.7 & 11.9 \\ \hline
\end{tabular}
\label{tab1}
\vspace*{3mm}

\centering
(b)~Protocol 2 with consensus error 0.3\\
\vspace*{1mm}
\begin{tabular}{|c||c|c|c|}
\hline
       & \multicolumn{2}{|c|}{Event-Triggered} 
       & Time- \\ \cline{2-3}
Graphs & $c_0=0.1$ & $c_0=0$ & Triggered \\ \hline
10 nodes  & 4.7 & 3.8 & 6.9 \\ \hline
50 nodes  & 5.9 & 5.6 & 8.1 \\ \hline
100 nodes & 6.2 & 6.5 & 8.4 \\ \hline
\end{tabular}
\label{tab2}
\vspace*{-3mm}
\end{table}




\vspace*{-5mm}

\begin{IEEEbiography}{Yuan Wang}
received the M.Sc.\ degree in engineering from Huazhong University 
of Science and Technology, Wuhan, China in 2016. 
He is currently pursuing the Ph.D.\ degree at the Department of 
Computer Science, Tokyo Institute of Technology, Yokohama, Japan. 
His main research interests are cyber-physical systems, event-based 
coordination, security in multi-agent systems, and model predictive 
control methods.
\end{IEEEbiography}

\vspace*{-5mm}

\begin{IEEEbiography}[
]{Hideaki Ishii} (M'02-SM'12) received the M.Eng.\ degree in
applied systems science from Kyoto University, Kyoto, Japan,
in 1998, and the Ph.D.\ degree in electrical and computer
engineering from the University of Toronto, Toronto, ON,
Canada, in 2002.
He was a Postdoctoral Research Associate with 
the Coordinated Science Laboratory at the University 
of Illinois at Urbana-Champaign, Urbana, IL, USA,
from 2001 to 2004, and a Research Associate with the
Department of Information Physics and Computing,
The University of Tokyo, Tokyo, Japan, from 2004 to 2007.
Currently, he is an Associate Professor in the Department 
of Computer Science, Tokyo Institute of Technology, Yokohama, Japan.
His research interests are in networked control systems,
multi-agent systems, hybrid systems, cyber security of power systems,
and probabilistic algorithms.

Dr.~Ishii has served as an Associate Editor for the 
\emph{IEEE Control Systems Letters}, 
and \emph{Mathematics of Control, Signals, and Systems}
and previously for \emph{Automatica},
the \emph{IEEE Transactions on Automatic Control}, and
the \emph{IEEE Transactions on Control of Network Systems}. 
He is the Chair of the IFAC Coordinating Committee on Systems and Signals
since 2017 and was the Chair of the IFAC Technical Committee on 
Networked Systems from 2011 to 2017. 
He received the IEEE Control Systems Magazine Outstanding
Paper Award in 2015.
\end{IEEEbiography}

\end{document}